\documentclass[reqno]{amsart}
\usepackage[utf8]{inputenc}
\usepackage{color}
\usepackage{textcomp}
\usepackage{url}
\usepackage{enumitem}
\usepackage{amstext}
\usepackage{amsthm}
\usepackage{amssymb}
\usepackage{graphicx}
\usepackage[unicode=true,pdfusetitle,
 bookmarks=true,bookmarksnumbered=false,bookmarksopen=false,
 breaklinks=false,pdfborder={0 0 0},pdfborderstyle={},backref=false,colorlinks=true]
 {hyperref}

\makeatletter

\newcommand*\LyXZeroWidthSpace{\hspace{0pt}}

\numberwithin{equation}{section}
\numberwithin{figure}{section}
\theoremstyle{plain}
\newtheorem{thm}{\protect\theoremname}
\theoremstyle{definition}
\newtheorem{defn}[thm]{\protect\definitionname}

\usepackage{pdfsync}
\usepackage{graphics}
\usepackage{epstopdf}
\usepackage[all]{xy}
\usepackage{amscd}
\setlist[enumerate]{leftmargin=*,label=(\roman*),align=left}
\setlist[itemize]{leftmargin=*,align=left}

\newcommand{\xyR}[1]{ \makeatletter
\xydef@\xymatrixrowsep@{#1} \makeatother} 
\newcommand{\xyC}[1]{ \makeatletter
\xydef@\xymatrixcolsep@{#1} \makeatother} 
\entrymodifiers={++[ ][F]} 

\newcommand{\ra}{\longrightarrow}
\newcommand{\xra}[1]{\xrightarrow{\ \ #1\ \ }} 

\newcommand{\field}[1]{\mathbb{#1}}
\newcommand{\R}{\field{R}} 
\newcommand{\N}{\field{N}} 

\renewcommand{\phi}{\varphi}
\newcommand{\eqUp}{\displaystyle \mathop {=}} 
\newcommand{\geUp}{\displaystyle \mathop {\ge}} 

\DeclareMathOperator{\act}{ac}
\DeclareMathOperator{\stateNeigh}{n}

\newcommand{\tst}{t_{\rm st}}
\newcommand{\tend}{t_{\rm end}}

\newcommand{\ag}[1]{{\rm ag}(#1)}
\newcommand{\pa}[1]{{\rm pa}(#1)}
\newcommand{\pr}[1]{{\rm pr}(#1)}
\newcommand{\tist}[1]{t_{#1}^{\rm s}}
\newcommand{\tong}[1]{t_{#1}^{\rm o}}
\newcommand{\tarr}[1]{t_{#1}^{\rm a}}

\newcommand{\tsubs}{t_{\rm s}}
\newcommand{\tsuba}{t_{\rm a}}
\newcommand{\tsubo}{t_{\rm o}}

\makeatother

\providecommand{\definitionname}{Definition}
\providecommand{\theoremname}{Theorem}

\begin{document}

\title[A mathematical definition of complex adaptive systems]{A mathematical definition of complex adaptive system as interaction
space}
\author{Paolo Giordano}
\address{Paolo Giordano, Faculty of Mathematics, University of Vienna, AT.}
\thanks{This research was funded in whole or in part by the Austrian Science
Fund (FWF) 10.55776/PAT9221023. For open access purposes, the author
has applied a CC BY public copyright license to any author-accepted
manuscript version arising from this submission}
\begin{abstract}
We define a mathematical notion of complex adaptive system by following
the original intuition of G.K.~Zipf about the \emph{principle of
least effort}, an intuitive idea which is nowadays informally widespread
in complex systems modeling. We call \emph{generalized evolution principle
}this mathematical notion of interaction spaces theory. Formalizing
and generalizing Mandelbrot's ideas, we also prove that a large class
of these systems satisfy a power law. We finally illustrate the notion
of complex adaptive system with theorems describing a Von Thünen-like
model. The latter can be easily generalized to other complex systems
and describes the appearance of emergent patterns. Every notion is
introduced both using an intuitive description with lots of examples,
and using a modern mathematical language.
\end{abstract}

\maketitle
\tableofcontents{}

\section{\label{sec:Complex-adaptive-systems}Complex adaptive systems and
emergent patterns following Zipf's idea}

It is widely recognized that complex adaptive systems (CAS) are among
the most interesting types of complex systems (CS) (see, e.g., \cite{Mi-Pa07,MiNe}),
and there are several attempts to provide a precise definition of
CAS, see e.g.~\cite{Ni-Pr77,Ba-Cr94,Hak04,Yu-So14}.

Due to its unifying properties, see \cite{Gio24_IS}, interaction
spaces (IS) theory can provide the appropriate context for a \emph{general
mathematical definition} of CAS, and the understanding of general
dynamical laws governing these systems. In the present paper, we use
all the notions of IS theory introduced in \cite{Gio24_IS}.

The main idea is to try a mathematical formalization of the original
ideas expressed by G.K.~Zipf in \cite{Zip49}, because they seem
intuitively clear, meaningful and applicable to a large class of CAS,
where some kind of optimization and, at the same time, of information
sharing between interacting entities in a CS is nowadays informally
frequently used, see e.g.~\cite{Bak97,Cru12,Gab99,Hak04,Man53,Mi-Pa07,JaBaCr,Ni-Pr77,Sor06}
and references therein.

In \cite{Zip49}, the appearing of a power law is connected to what
we think is an adapting behavior of the system. Zipf's \emph{principle
of least effort} intuitively explains this adaptation as a result
of \emph{two} opposing processes: \emph{unification} and \emph{diversification}.
We will try to intuitively explain this principle through several
examples listed below. Clearly, both these terms in \cite{Zip49}
are sufficiently imprecise and could hence be misinterpreted. The
mathematical formalization we present here also aims to gain a more
clear understanding of these processes.

In our insight, unification processes are related to interactions
that try to \emph{decrease} convenient \emph{costs} (possibly meant
in an abstract way: e.g.~loss of profits, loss of common goods, probability
to get hungry, probability to loose reproductive possibilities, hormones
such as cortisol or dopamine, etc.), whereas diversification processes
are linked to \emph{long term changes} of suitable interactions, i.e.~to
the increasing of the \emph{causal information shared through the
goods} exchanged between agents and patients of these interactions,
see \cite{Gio24_IS} for an explanation of these terms of IS theory.
It is the implementation of these interactions and the most diversified
exchange of fluxes of goods that enable the adaptive population to
be resilient \emph{and} keep a low value of costs.

To start a first understanding and a preliminary intuitive validation
of the subsequent mathematical definition, we can keep in mind the
following examples:
\begin{enumerate}[label=\arabic*)]
\item In a natural language, unification processes drift to shorten most
frequently used words (or better: frequently used sounds, see \cite{BHGK12});
diversification ones make evolve the language towards longer and specialized
words, \cite{CoMi,Zip49}.
\item In cities and their markets development, unification tends to bring
near people so as to decrease suitable costs of living; diversification
tends to use all the possible living locations so as to approach the
appropriate rent costs, \cite{Gab99}.
\item In natural selection, unification forces push giraffes to search for
eatable trees (see e.g.~\cite{CadT07} and references therein); diversification
can select all the best genetic codes that allow for a longer neck,
\cite{Dan-etal15}. We recall that a costs decreasing process (unification)
can cause an evolution into different phenotypes (diversification),
see \cite{HeDo13}. More generally, natural selection seems to result
by the dynamics of two processes: representation of biological information
as chemical properties and control flow (diversification), and the
energetic constraints limiting the maintenance of that information
(unification), \cite{Smi-I,Smi-II,Smi-III}. See \cite{SimSch96}
for an alternative evolutionary explanation of long neck in giraffes.
\item Determining the direction to navigate to a safe place, such as a home
or nest, is a fundamental behavior for all complex animals and a crucial
first step in navigation. We can say that unification mental processes
evaluate costs related to the achievement of these goals, whereas
diversification ones are related to judgments based on a previously
learned or planned behavior, \cite{Cha-etal15}.
\item Companies with a longer life span are able not only to decrease costs
and increase profits (unification), but are also able to adapt to
their complex environment by implementing long-term robustness. The
latter are often realized through diversification processes such as:
maintain heterogeneity of people, ideas, and endeavors, and preserve
redundancy among components, \cite{ReLeUe16}.
\item In autism (but also in schizophrenia), a necessary decreasing in high
costs related to social interactions (unification) is sometimes compensated
by higher abilities (diversification) in very specialized or creative
activities, \cite{Pow15,CrStEl10}.
\item The ability to manage costs related to large varieties of goods (unification)
is related to the ability to implement the same stable economic decisions
(interactions) applicable to different goods (diversification), \cite{Xie-Pad16}.
\item Phyllotaxis, the regular arrangement of leaves or flowers around a
plant stem, is an example of developmental pattern formation. Phyllotaxis
is characterized by the divergence angles between the organs, the
most common angle being 137.5°, the golden angle. Different approaches
and hypotheses has been used to model this formation mechanism, see
e.g.~\cite{Kuh07}. In this process, we can see unification forces
related to energy exploitation by each primordium, and diversification
forces that tend to uniformly distribute these energy sources between
old and new primordia.
\item \label{enu:WS}An example of \textbf{non}-adaptive but still complex
system is traders payments of Wall Street employees. It is well known
that this payment has a base salary and a bonus, which is usually
a percentage of trader’s profit. If traders lose, they still get their
base, and only if their loss is great enough, they are fired. However,
they never have to return the money lost by the company due to their
wrong trading. This is a clear financial incentive to be reckless
because it rewards short-term gains (costs which identify unification
forces) without regard to long-term consequences (diversification
forces), \cite{Arn13}.
\item The efficiency of a parliament (unification processes interpreted
as decreasing of suitable costs) can be improved by inserting randomly
selected legislators (increasing of diversification among legislators),
see \cite{PGRSC11}.
\item Whereas classical economic theories prescribe specialization of countries
industrial production, inspection of the country databases of exported
products shows that this is not the case: successful countries are
extremely diversified, in analogy with biosystems evolving in a competitive
dynamical environment. In fact, together with classical and necessary
costs reduction (unification), diversification represents the hidden
potential for development and growth, \cite{TCCGP12}.
\item The evolutionary emergence of an egalitarian attitude in a population
can be explained by using an evolutionary model of group-living individuals
competing for resources and reproductive success (i.e.~unification
as costs to decrease), see \cite{Gav12}. Although the differences
in fighting abilities lead to the emergence of hierarchies where stronger
individuals take away resources from weaker individuals, the logic
of within-group competition implies that each individual benefits
if the transfer of resources from a weaker group member to a stronger
one is prevented. This model shows that this effect can result in
the evolution of a particular behavior causing individuals to interfere
in a bully–victim conflict on the side of the victim. A necessary
condition in this process is a high efficiency of coalitions in conflicts
against the bullies. The egalitarian drive leads to a dramatic reduction
in within-group inequality. Simultaneously, it creates the conditions
for the emergence of inequity aversion via the internalization of
these adapted behavioral rules. All these interactions are general
because they can be applied to different situations. They hence represent
a long-term and diversified improving of society.
\item \label{enu:Instit}The present network of financial exposures among
institutions (e.g.~banks, companies) shows that this system is \textbf{not}
well adapted. The centrality of certain institutions does not allow
the system to be resilient to financial fails of few institutions,
see \cite{BPKTC12}. As we will see later, this is the same mechanism
showed by monopolies and the corresponding lack of diversity in the
system. We could say that the system is flawed because it lacks in
interactions that act when a node of the financial network fail: a
more diversified network would show a greater resilience because of
a lower centrality of these nodes. It is also clear that a globally
directed taxation system can prevent the formation of such a centralized
nodes. Even if these global taxes do not allow a maximization of profits
(which could be clearly higher without this tax), they could be rightly
justified as a measure to prevent global irreparable problems to the
entire system, and hence to the institutions themselves. See also
\cite{Vic10}.
\end{enumerate}
To underline that we are going to give our interpretation of Zipf's
principle of least effort, we prefer to name our mathematical version
\emph{generalized evolution principle} (GEP)\emph{.} This name has
also the merit to link this adapting dynamics to evolutionary theories
that, in our opinion, are well inscribed into it.

Before starting an intuitive description of the GEP, we have first
to specify what is a global state of an IS.

\section{\label{sec:Global-state}Global states of an interaction space}

There are five sources of (possible) randomness acting in an IS $\mathfrak{I}$.
One of them derives from the (possible) stochastic evolution equation,
see \cite[(EE)]{Gio24_IS}. We recall that this also includes the
stochastic behaviour of activation state $\act_{i}^{p}(t)\in[0,1]$
and goods $\gamma_{i}(t)\in R_{i}$, see \cite[Rem.~1~(e)]{Gio24_IS}.
We have this type of stochastic behavior because we do not want to
model some details of the evolution of the interaction $i$, or because
it is not possible at all: it can happens, e.g., that a pedestrian
randomly chooses between two exists which are equally located with
respect to its present position: a deterministic model could not be
reasonable for a non-trivial large class of pedestrians.

Three other sources of randomness are the occurrence times $\tist{i}(t)$,
$\tarr{i}(t)$ and $\tong{i}(t)$, see \cite[Sec.~3.5.1]{Gio24_IS}.

The last one, is due to neighborhoods $\mathcal{N}_{i}(t)\subseteq E$
of interactions $i\in I$, see \cite[(NE)]{Gio24_IS}.

Therefore, in principle, it is possible to define a new probability
space that combines all these (possible) random sources. Elementary
events of this space are of the form $(\tsubs,\tsuba,\tsubo,\mathcal{N},\omega)$,
where $(\tsubs,\tsuba,\tsubo)\in[\tst,\tend]^{3}$ are e.g.~distributed
as explained in \cite[Sec.~3.5.1]{Gio24_IS}, $\mathcal{N}\subseteq E$
are all possible neighbourhoods, probabilistically distributed depending
on modeling choices, and $\omega\in\Omega_{p}$ is distributed as
$P_{p}$ for all patients $p$ according to the evolution equation,
see \cite[(EE)]{Gio24_IS}.

The joint probability that we have to settle on this space of elementary
events is certainly not easy to set. On the one hand, the universal
properties of IS allow us to state that the aforementioned ones are
all the possible sources of randomness that we have to take into account
in several type of models. On the other hand, the causal graph defined
by the activation function, \cite[Sec.~2]{Gio24_IS}, can be of great
help in finding this probability: we can say that the interaction
$i$ is a cause of the interaction $j$ if along any possible solution
of the IS the interaction $i$ always precedes $j$, and if $i$ activates
for $j$ an agent of $j$ (see e.g.~Fig.~\ref{fig:exampleCause-effect}).

\begin{figure}
\noindent \begin{centering}
\includegraphics[scale=0.75]{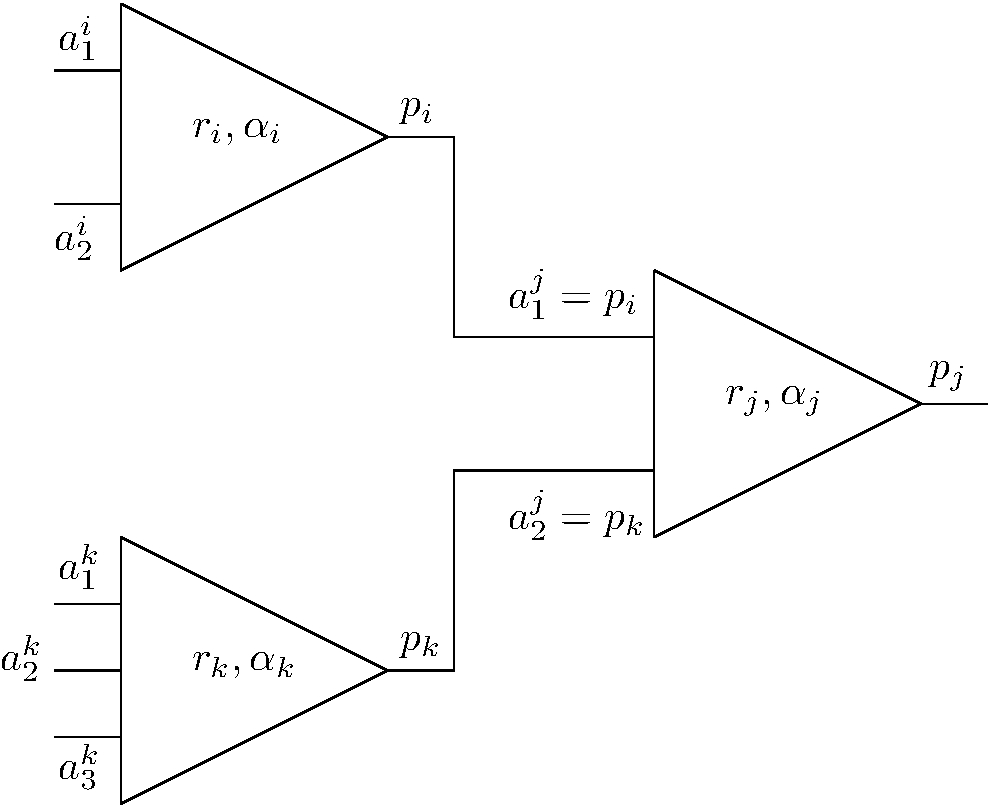}
\par\end{centering}
\caption{\label{fig:exampleCause-effect}An example of interactions in cause-effect
cascade.}
\end{figure}
If the resulting directed cause-effect graph is acyclic, we can interpret
it as a beliefs network and apply the methods of Bayesian networks
to define the joint probability, see e.g.~\cite{Dar09}.

A\emph{ global state space $\bar{M}_{\mathcal{P},J,t}$ of the population
$\mathcal{P}\subseteq E$ and the interactions $J\subseteq I$ up
to time $t\in[\tst,\tend]$ }is given by three components, each being
a subset of all the possible \emph{paths} (states, times, neighbourhoods)
of our IS (recall that $Y^{X}$ is the set of all the functions $f:X\ra Y$
and, to understand, that if the index set $J=\{j_{1},\ldots,j_{n}\}$
is finite, then the product of sets is $\prod_{j\in J}S_{j}=S_{j_{1}}\times\ldots\times S_{j_{n}}$):
\begin{align}
\bar{M}_{\mathcal{P},J,t} & :=\bar{M}_{\mathcal{P},t}^{\text{s}}\times M_{J,t}^{\text{t}}\times M_{J,t}^{\text{n}}\label{eq:globalStatesAll}\\
\bar{M}_{\mathcal{P},t}^{\text{s}} & \subseteq\left(\prod_{e\in\mathcal{P}}\bar{S}_{e}\right)^{[\tst,t]}=:\bar{S}_{\mathcal{P}}^{[\tst,t]}\label{eq:globalStates}\\
M_{J,t}^{\text{t}} & \subseteq\left([\tst,\tend]^{3}\right)^{J\times[\tst,t]}\nonumber \\
M_{J,t}^{\text{n}} & \subseteq\left\{ \mathcal{N}\mid\mathcal{N}\subseteq E\right\} ^{J\times[\tst,t]}\nonumber \\
\bar{M}_{\mathcal{P},J} & :=\bigcup_{t\in[\tst,\tend]}\bar{M}_{\mathcal{P},J,t}.\label{eq:globalStateSpace}
\end{align}
The space $\bar{M}_{\mathcal{P},J}$ is called \emph{global state
space of $\mathcal{P}$ and $J$.}

Recall that the state space $\bar{S}_{e}$ also contains activation
and states of goods, whereas the proper state space $S_{e}$ is a
subspace of $\bar{S}_{e}$ that includes all the other proper state
values, see \cite[(ST)]{Gio24_IS}. The optimization performed by
a CAS occurs in the space $\bar{M}_{\mathcal{P},J,t}$ and hence it
also depends on the choice of such a global state space, see also
Sec.~\ref{subsec:General-steps-to-GEP}.

A conceptual explanation of supersets appearing in \eqref{eq:globalStates}
is as follows: for each interacting entity $e\in\mathcal{P}$ of the
given population, in $\bar{M}_{\mathcal{P},t}^{\text{s}}$ we have
all the state \emph{time functions} $x_{e}:[\tst,t]\ra\bar{S}_{e}$
of our IS; in $M_{J,t}^{\text{t}}$ we can consider all the occurrence
times $\tist{j}(-)$, $\tarr{j}(-)$, $\tong{j}(-):[\tst,t]\ra[\tst,\tend]$
of each interaction $j\in J$; finally, in $M_{J,t}^{\text{n}}$ we
have the neighbourhood $\mathcal{N}_{j}(-):[\tst,t]\ra\left\{ \mathcal{N}\mid\mathcal{N}\subseteq E\right\} $
of $j\in J$ as function of time up to $t$.

Even if it is natural to consider in the global state space all the
possible paths of independent variables we can consider in our IS
model, in the mathematical definition of the GEP Def.~\ref{def:CAS}
we will see why this is important to achieve a greater generality
and flexibility in defining from what the adaptation property must
depend on. The intuitive idea preliminary expressed in the global
space $\bar{M}_{\mathcal{P},J,t}$ is that a population $\mathcal{P}$
can adapt by changing in time its proper state variables, or its activations,
exchanged goods (and hence its cause-effect relations), occurrence
times or even neighbourhoods where interactions $j\in J$ take information
they need to run. We simply use the same symbols but omitting the
subscripts $\mathcal{P}$ and $J$ (e.g.~$\bar{M}_{t}$, $\bar{S}$,
$\bar{M}_{t}^{\text{s}}$, etc.) if $\mathcal{P}=E$ and $J=I$.

The idea to consider only the time interval $[\tst,t]$ up to $t$,
allows us to mathematically define that the system is better adapted
at time $t$ than at time $s$, and hence to distinguish between an
improving dynamic and an emergent pattern, which is a best possible
global state.

On the basis of what we said above, we can assume to have a probability
space $(\Omega^{\text{g}},\mathcal{F}^{\text{g}},P^{\text{g}})$ (the
superscript ``$\text{g}$'' stands for \emph{global}) and three
random processes
\begin{align}
X & :\Omega^{\text{g}}\times[\tst,\tend]\ra\bar{S}\nonumber \\
T & :\Omega^{\text{g}}\times I\times[\tst,\tend]\ra[\tst,\tend]^{3}\label{eq:randVarIS}\\
N & :\Omega^{\text{g}}\times I\times[\tst,\tend]\ra\left\{ \mathcal{N}\mid\mathcal{N}\subseteq E\right\} \nonumber 
\end{align}
representing resp.~the possible state maps of each interacting entity
$e\in E$ (including activations, goods and proper state variables),
the possible occurrence times, and the neighbourhood of each interaction
$i\in I$. In other words, for any fixed elementary event $w\in\Omega^{\text{g}}$,
if we replace everywhere the state map $x_{e}(t)\in\bar{S}_{e}$,
the occurrence times $\tist{i}(t)$, $\tarr{i}(t)$, $\tong{i}(t)\in[\tst,\tend]$,
and the neighborhoods $\mathcal{N}_{i}(t)\subseteq E$ resp.~with
the time functions $X(w)(t)_{e}:=X(\omega,t)_{e}\in\bar{S}_{e}$,
$T_{i}^{\text{s}}(\omega)(t):=T(\omega,i,t)_{1}$, $T_{i}^{\text{a}}(\omega)(t):=T(\omega,i,t)_{2}$,
$T_{i}^{\text{o}}(\omega)(t):=T(\omega,i,t)_{3}$, and $N(\omega,i,t)\subseteq E$,
then all the conditions in the definition of IS are satisfied. For
example, $X(w)(t)_{e}$ satisfies the evolution equation whenever
$e$ is a patient fulfilling the assumptions of \cite[(EE)]{Gio24_IS}.

The process $X=X(\omega,t)_{e}$ is formally a function of three variables
$\omega\in\Omega^{\text{g}}$, $t\in[\tst,\tend]$ and $e\in E$,
but we use flexible notations for the corresponding partial functions
when one or more of these variables are fixed. For example, $X_{e}(t):\omega\in\Omega^{\text{g}}\mapsto X(\omega)(t)_{e}\in\bar{S}_{e}$
is the random state variable of the entity $e\in\mathcal{P}$ at time
$t\in[\tst,\tend]$. Similar notations are also used for the processes
$T$ and $N$.

Choosing a space $\bar{M}_{\mathcal{P},t}^{\text{s}}$ which contains
only constant state maps, we are equivalently considering the adaptation
as occurring in a subspace (e.g.~a manifold) of $\bar{S}_{\mathcal{P}}=\prod_{e\in\mathcal{P}}\bar{S}_{e}$;
this is what we will examine both in Sec.~\ref{sec:Power-law-distribution}
and in Sec.~\ref{subsec:Generalizing-Von-Thunen} below. On the contrary,
an example of IS where the time state function $X(\omega)(-)\in\bar{M}_{\mathcal{P},\tend}^{\text{s}}$
is more important than the value itself $X(\omega)(t)_{e}\in\bar{S}_{e}$
is in the intelligent interpretation of a given text, let us say a
clinical note, where we want to count how much frequently a given
disease $e$ is cited in the text: $X(\omega)(t_{1})_{e},\ldots,X(\omega)(t_{N})_{e}$.
In this example, the variable $t$ represents the passing time while
a given ``user $\omega$'' is reading the text. If we consider only
$t_{1},\ldots,t_{N}\le t$, then $X(\omega)(-)_{e}|_{[\tst,t]}$ can
be use to represent the amount of information the user $\omega$ is
interpreting in the text up to reading time $t$.

In the following, we also always assume that the global state space
$\bar{M}_{\mathcal{P},J,t}$ is chosen so that
\begin{align}
\left(X_{e}(\omega)|_{[\tst,t]}\right)_{e\in\mathcal{P}} & \in\bar{M}_{\mathcal{P},t}^{\text{s}}\nonumber \\
\left(T_{i}(\omega)|_{[\tst,t]}\right)_{i\in J} & \in M_{J,t}^{\text{t}}\qquad\forall\omega\in\Omega^{\text{g}}\,\forall t\in[\tst,\tend]\label{eq:indepPopInt}\\
\left(N_{i}(\omega)|_{[\tst,t]}\right)_{i\in J} & \in M_{J,t}^{\text{n}}\nonumber 
\end{align}
for all populations $\mathcal{P}\subseteq E$ and all families of
interactions $J\subseteq I$. 

Finally, using $X(\omega)(t)_{e}\in\bar{S}_{e}$, we can also define
the random processes of \emph{activation state} and \emph{goods} for
each interaction $i\in I$:
\begin{align*}
\text{AC}_{i}^{e}(\omega)(t) & :=X(\omega)(t)_{e,1,i}\\
\Gamma_{i}(\omega)(t) & :=X(\omega)(t)_{\pr{i},2,i},
\end{align*}
so that $\text{AC}_{i}^{e}(t):\Omega^{\text{g}}\ra[0,1]$ and $\Gamma_{i}(t):\Omega^{\text{g}}\ra R_{i}$.

\section{\label{sec:Intuitive-description-GEP}Intuitive description of the
generalized evolution principle}

We start with an intuitive description of the GEP and with some thoughts
to understand what we have to define in a precise mathematical language.

In an IS, we can have several populations of interacting entities.
Since only some of these populations have to be described as adapting,
we have to talk of the \emph{adaptation of a given} subset (=:\emph{
population}) \emph{of interacting entities}. In other words, it is
not necessarily correct to talk of an entire ``complex adaptive IS''.

Therefore, in an adaptation process, we need to identify an adapting
population $\mathcal{P}$ of interacting entities, $\mathcal{P}\subseteq E$,
and a family of interactions $I_{\mathcal{P}}\subseteq I$ performed
by the population $\mathcal{P}$ and representing how $\mathcal{P}$
is going to \emph{decreases costs and stabilize the adaptation process
through a suitable diversification of its goods}. Interactions in
$I_{\mathcal{P}}$ are called \emph{adaptive interactions}.

\noindent In the following definition, we specify what kind of interactions
we have to consider:
\begin{defn}
\label{def:popInteraction}Let $\mathcal{P}\subseteq E$, we say that
$I_{\mathcal{P}}\subseteq I$ \emph{is a family of interactions of
(the population) }$\mathcal{P}$\emph{ }if at least one agent of each
interaction $i\in I_{\mathcal{P}}$ is a member of the population
$\mathcal{P}$:
\[
\forall i\in I_{\mathcal{P}}\,\exists a\in\ag{i}\cap\mathcal{P}.
\]
\end{defn}

\noindent We now continue the intuitive motivations and description
of the GEP.

Note explicitly that not necessarily interactions of $I_{\mathcal{P}}$
act on entities of the population $\mathcal{P}$, i.e.~not always
$\pa{i}\in\mathcal{P}$ if $i\in I_{\mathcal{P}}$. This opens the
possibility that the population adapts by acting on external entities
and, only after a suitable chain of cause-effect related interactions,
this causes a change in the state of $\mathcal{P}$.

As we will see more precisely below, the adaptation process changes
functions state $X_{\mathcal{P}}:=\left(X_{e}|_{[\tst,s]}\right)_{e\in\mathcal{P}}$
of the population $\mathcal{P}$, or occurrence times $T_{\mathcal{P}}(s):=\left(T_{i}|_{[\tst,s]}\right)_{i\in I_{\mathcal{P}}}$
or neighborhoods $N_{\mathcal{P}}(s):=\left(N_{i}|_{[\tst,s]}\right)_{i\in I_{\mathcal{P}}}$
of its adaptive interactions at time $s$ to a better value at time
$t$, where unification and diversification are improved, i.e.~costs
are lowered and goods are resiliently distributed. We therefore precisely
define when it happens that \emph{at time $t$ the population $\mathcal{P}$
is better adapted than at time $s$}. A particular case will be when
at time $t$ the population is at one of the \emph{best} possible
states, which usually corresponds to a steady or an equilibrium state
and hence to an emergent pattern, even if, in general, not necessarily
an emergent pattern can be really attained as actual state.\\
In the following, we use the notation for the global state
\begin{equation}
Y_{\mathcal{P}}(\omega)|_{[\tst,t]}:=\left(\left(X_{e}(\omega)|_{[\tst,t]}\right)_{e\in\mathcal{P}},\left(T_{i}(\omega)|_{[\tst,t]}\right)_{i\in I_{\mathcal{P}}},\left(N_{i}(\omega)|_{[\tst,t]}\right)_{i\in I_{\mathcal{P}}}\right),\label{eq:Y-state}
\end{equation}
for $\omega\in\Omega^{\text{g}}$ and $t\in[\tst,\tend]$. Note that
$Y_{\mathcal{P}}(\omega)|_{[\tst,t]}\in\bar{M}_{\mathcal{P},I_{\mathcal{P}},t}\subseteq\bar{M}_{\mathcal{P},I_{\mathcal{P}}}$
because of \eqref{eq:indepPopInt}.

Together with the family of interactions $I_{\mathcal{P}}$, for each
global state $y\in\bar{M}_{\mathcal{P},I_{\mathcal{P}}}$, we have
to identify a function
\begin{equation}
C_{y}:\mathcal{P}\ra\R_{\ge0}^{k}\qquad\forall y\in:\bar{M}_{\mathcal{P},I_{\mathcal{P}}}\label{eq:costFunction}
\end{equation}
defined on the population $\mathcal{P}$ and called \emph{unification
costs} \emph{causing }$I_{\mathcal{P}}$\emph{ at the state $y$}:
an adaptive population $\mathcal{P}$ reacts with the interactions
$I_{\mathcal{P}}$ to an increase in these costs $C_{y}=(C_{y}^{1},\ldots,C_{y}^{k})\in\R_{\ge0}^{k}$
by changing the global state $y$ and trying to decrease each cost
$C_{y}^{j}$, $j=1,\ldots,k$. The interactions $i\in I_{\mathcal{P}}$
that allow a decreasing of each component of the cost function represent
the unification processes. Therefore, we are going to define when
$\mathcal{P}$ is adaptive \emph{with respect to a given cost function
$C$ and a family of interactions $I_{\mathcal{P}}$}. Of course,
a population $\mathcal{P}$ can be adaptive with respect to several
cost functions and families of interactions. Note that since the cost
function $C_{y}$ depends on the global state space $y\in\bar{M}_{\mathcal{P},I_{\mathcal{P}},t}$
(the union in \eqref{eq:globalStateSpace} is actually a disjoint
one), we can also consider costs depending on the activations or proper
state variables of $e\in\mathcal{P}$, or on goods, occurrence times
or neighbourhoods of interactions $i\in I_{\mathcal{P}}$, or even
on time $t$.

In general, the adaptation of the population $\mathcal{P}$ does not
occur by changing the state of a single entity $e\in\mathcal{P}$
but of several of them. For this reason, in most cases $C_{y}$ is
actually evaluated by averaging costs of the same type undergone by
each interacting entity. The probability we use to average these costs
is defined on the population $\mathcal{P}$, so that we can talk of
costs paid \emph{by the population} $\mathcal{P}$ with respect to
this probability. In general, this probability also depends on the
global state $y\in\bar{M}_{\mathcal{P},I_{\mathcal{P}}}$: think e.g.~at
the following examples:
\begin{enumerate}[label=\arabic*)]
\item $y=(y^{1},\ldots,y^{n})\in\R^{n}$ are the frequencies $y^{k}\in\R$
of each of $n$ words in an evolving language. If the language is
a CAS, these frequencies are distributed following a power law, and
the average values of the costs in using this language depends on
these frequencies, giving a higher weights to more used words, see
Sec.~\ref{subsec:ExamplesCostsFunct}
\item $y=(x,r)\in\R^{2n}\times\R_{\ge}^{n}$ are locations $x^{j}\in\R^{2}$
and rents $r^{j}\in\R_{\ge0}$ of $j=1,\ldots,n$ companies producing
suitable commodities, and the costs of these companies are averaged
depending on both locations and rents, e.g.~giving higher weights
to more central and expansive locations, see Sec.~\ref{subsec:Generalizing-Von-Thunen}.
\item $y$ is the degree of attention (i.e.~the activation in the language
of IS) that each lecturer $\omega$ is giving to each part of speech
in a given clinical note during its reading. The most adapted readings
we are interested to are those where we have a probability which are
proportional to the degree of attention given to medical terms, and
hence to a lower value of the related average costs.
\end{enumerate}
Therefore, before talking of an adaptation process, for each global
state $y\in\bar{M}_{\mathcal{P},I_{\mathcal{P}}}$ and for each cost
$C_{y}^{j}$, we also need to identify a probability $P_{y}^{j}$
on \emph{$\mathcal{P}$}:
\[
P_{y}^{j}\text{ is a probability on }\mathcal{P}\qquad\forall y\in\bar{M}_{\mathcal{P},I_{\mathcal{P}}}\,\forall j=1,\ldots,k.
\]
It is with respect to $P_{y}^{j}$ that we have to evaluate the expected
value of the cost function $C_{y}^{j}$, and it is this expected value,
and not the single costs experienced by each $e\in\mathcal{P}$, that
have to be decreased in unification interactions. The probability
$P_{y}^{j}$ is called \emph{probability to average unification cost
$C_{y}^{j}$}.

We therefore define the \emph{unification ``forces'' at the state
$y$} by
\begin{equation}
U_{\mathcal{P}}^{j}(y):=-E_{y}^{j}\left[C_{y}^{j}\right]\qquad\forall y\in\bar{M}_{\mathcal{P},I_{\mathcal{P}}}\,\forall j=1,\ldots,k,\label{eq:unificationForcesState}
\end{equation}
where the expected value $E_{y}^{j}[-]$ is computed using $P_{y}^{j}$.
The minus sign in \eqref{eq:unificationForcesState} allows one to
say that these forces are greater when the average costs are lower.
For example, if $\mathcal{P}=\{e_{1},\ldots,e_{N}\}$ is finite, then
\begin{equation}
U_{\mathcal{P}}^{j}(y)=-\sum_{e\in\mathcal{P}}P_{y}^{j}(e)\cdot C_{y}^{j}(e).\label{eq:UnificationFinite}
\end{equation}
Similarly, we can define the \emph{unification forces for the event
$\omega\in\Omega^{\text{g}}$ at time $t\in[\tst,\tend]$} by 
\begin{equation}
U_{\mathcal{P}}^{j}(\omega,t):=U_{\mathcal{P}}^{j}(Y_{\mathcal{P}}(\omega)|_{[\tst,t]}).\label{eq:unificationForcesSample}
\end{equation}
Therefore, if the unification interactions allows the system to pass
from the global state $Y_{\mathcal{P}}(\omega)|_{[\tst,s]}$ to the
better state $Y_{\mathcal{P}}(\omega)|_{[\tst,t]}$, this can be expressed
with
\begin{align}
U_{\mathcal{P}}^{j}(\omega,t) & \ge U_{\mathcal{P}}^{j}(\omega,s)\quad\forall j=1,\ldots,k.\label{eq:greaterUnifForcesTimes}
\end{align}
In addition, the condition that $\bar{y}\in\bar{M}_{\mathcal{P},I_{\mathcal{P}}}$
is the best global state from the point of view of costs, can be expressed
with
\begin{equation}
U_{\mathcal{P}}^{j}(\bar{y})\ge U_{\mathcal{P}}^{j}(\omega,t)\quad\forall\omega\in\Omega^{\text{g}}\,\forall t\in[\tst,\tend]\,\forall j=1,\ldots,k.\label{eq:greaterUnifForcesState}
\end{equation}

The adaptive interactions $i\in I_{\mathcal{P}}$ also realize the
diversification to a decreasing of $C_{y}$, and hence to stabilize
this declining of costs. We measure the \emph{diversification ``forces''}
$D_{I_{\mathcal{P}}}(\gamma)$ with the \emph{information entropy
of the fluxes of goods $\gamma_{i}\in R_{i}$ extracted by each population
interaction $i\in I_{\mathcal{P}}$ from its resource space} $R_{i}$,\emph{
and exchanged from agents to patients through propagators}, see \cite{Gio24_IS}
and \cite{Cru12,BaCr15} for a similar point of view. An intuitive
way to motivate this idea is to say that the more the population is
sharing its goods/resources/information the more it is adapting. See
also Sec.~\ref{sec:GEP-and-others} below for the relationships between
GEP and different approaches in measuring information flows.

From what space these goods $\gamma=(\gamma_{i})_{i\in I_{\mathcal{P}}}$
are taken from? If $x\in\bar{M}_{\mathcal{P},t}^{\text{s}}$ is a
state map of our IS, then $\gamma_{x,i}:=x(-)_{\pr{i},i,2}:[\tst,t]\ra R_{i}$
are the corresponding goods of $i\in I$ as a function of time (up
to $t$). It suffices to consider the family $\left(\gamma_{x,i}\right)_{i\in I_{\mathcal{P}}}$;
e.g.~if $I_{\mathcal{P}}=\{i_{1},\ldots,i_{d}\}$ is finite, then
$\left(\gamma_{x,i}\right)_{i\in I_{\mathcal{P}}}=\left(\gamma_{x,i_{1}},\ldots,\gamma_{x,i_{d}}\right)$
considers all the goods extracted by each interaction in the given
order. The spaces we need to consider are therefore
\begin{align}
R_{\mathcal{P},t} & :=\left\{ \left(\gamma_{x,i}\right)_{i\in I_{\mathcal{P}}}\mid x\in\bar{M}_{\mathcal{P},t}^{\text{s}}\right\} \nonumber \\
R_{\mathcal{P}} & :=\bigcup_{t\in[\tst,\tend]}R_{\mathcal{P},t}.\label{eq:spacesGoods}
\end{align}

We therefore need to finally identify a probability $Q_{\gamma}$
defined on $I_{\mathcal{P}}$ to evaluate the corresponding forces
of diversification $D_{I_{\mathcal{P}}}(\gamma)$ and $D_{I_{\mathcal{P}}}(\omega,t)$
as
\begin{align}
Q_{\gamma} & \text{ is a probability on }I_{\mathcal{P}}\qquad\forall\gamma\in R_{\mathcal{P}}\label{eq:diversificationForces}\\
D_{I_{\mathcal{P}}}(\gamma) & :=\text{Entropy}\left(Q_{\gamma}\right)\\
D_{I_{\mathcal{P}}}(\omega,t) & :=D_{I_{\mathcal{P}}}\left(\left(\Gamma_{i}(\omega)|_{[\tst,t]}\right)_{i\in I_{\mathcal{P}}}\right)\qquad\forall\omega\in\Omega^{\text{g}}\,\forall t\in[\tst,\tend].
\end{align}
The probability $Q_{\gamma}$ is called \emph{diversification probability}.
For example, if $I_{\mathcal{P}}=\{i_{1},\ldots,i_{d}\}$ is finite,
then
\begin{equation}
D_{I_{\mathcal{P}}}(\gamma)=-\sum_{i\in I_{\mathcal{P}}}Q_{\gamma}(i)\cdot\log_{2}Q_{\gamma}(i).\label{eq:DiversificationFinite}
\end{equation}
Note that, in general, the diversification forces depend on the whole
history $\stateNeigh_{\pr{i}}x_{t}$ of the state of the neighborhood
of the propagator of $i\in I_{\mathcal{P}}$, because, e.g.~for the
case of $D_{I_{\mathcal{P}}}(\omega,t)$, the state of goods $\Gamma_{i}(t)(\omega)$
satisfy the evolution equation \cite[(EE)]{Gio24_IS}, see also \cite[Rem.~1.(e)]{Gio24_IS}.
Therefore, this neighborhood can include the influence of entities
which are external to $\mathcal{P}$. Similarly, the propagator of
$i$ not necessarily belongs to $\mathcal{P}$, so that the diversification
of resources can also involve external entities.

The second important condition of the GEP states that \emph{in the
sample $\omega\in\Omega^{\text{g}}$} \emph{diversification forces
are greater at time }$t$ \emph{than at time }$s$ if
\begin{equation}
D_{I_{\mathcal{P}}}(\omega,t)\ge D_{I_{\mathcal{P}}}(\omega,s).\label{eq:greaterDiversTimes}
\end{equation}
As above, the condition that the global state of goods $\bar{\gamma}\in R_{\mathcal{P}}$
is the best from the point of view of diversification forces can be
stated asking that
\begin{equation}
D_{I_{\mathcal{P}}}(\bar{\gamma})\ge D_{I_{\mathcal{P}}}(\omega,t)\qquad\forall\omega\in\Omega^{\text{g}}\,\forall t\in[\tst,\tend].\label{eq:greaterDiversState}
\end{equation}
For example, assume that we have two families of simultaneous \emph{independent}
interactions $J_{\mathcal{P}}$, $K_{\mathcal{P}}$ for the probability
$Q(\gamma)$ (e.g.~because there is no cause-effect relation between
$j\in J_{\mathcal{P}}$ and $k\in K_{\mathcal{P}}$). Then the diversifications
$D_{J_{\mathcal{P}}}(\omega,t):=D_{J_{\mathcal{P}}}\left(\left(\Gamma_{i}(\omega)|_{[\tst,t]}\right)_{i\in J_{\mathcal{P}}}\right)$
and $D_{K_{\mathcal{P}}}(\omega,t):=D_{K_{\mathcal{P}}}\left(\left(\Gamma_{i}(\omega)|_{[\tst,t]}\right)_{i\in K_{\mathcal{P}}}\right)$,
thanks to the logarithm in \eqref{eq:diversificationForces}, will
contribute additively $D_{I_{\mathcal{P}}}(\omega,t)=D_{J_{\mathcal{P}}}(\omega,t)+D_{K_{\mathcal{P}}}(\omega,t)$
at increasing the diversification of the population $\mathcal{P}$.

In the GEP, we ask both \eqref{eq:greaterUnifForcesTimes} and \eqref{eq:greaterDiversTimes}
or both \eqref{eq:greaterUnifForcesState} and \eqref{eq:greaterDiversState}
(but where $\bar{\gamma}$ are the goods already included in the global
state $\bar{y}$), see below Def.~\ref{def:CAS}.

\section{\label{sec:GEP-and-others}Generalized evolution principle, Shannon
entropy, out of equilibrium systems and second law of thermodynamics}

We recall that, in information theory, if two different messages are
extracted from the same probability distribution, then they are undistinguishable
from the point of view of information entropy. Actually, it is not
correct to talk of Shannon entropy of single messages, because this
notion can be applied only to probability distributions (and hence,
e.g., to random variable). This is frequently summarized saying that
entropy does not depend on the meaning of messages but only on their
probability distribution. On the contrary, if we have two messages,
and the former at time $t$ (let us say \emph{``to be or not to be}'')
exchanges greater fluxes of goods with the population $\mathcal{P}$
than the latter (e.g.~``\emph{nttb e obt ooo e r}'') at time $s$,
then inequality \eqref{eq:greaterDiversTimes} could be interpreted
saying that the message exchanged at $t$ is more meaningful for $\mathcal{P}$
than the message exchanged at time $s$. Therefore, the GEP has a
different meaning than the simple Shannon entropy, exactly because
it involves interactions and their propagators, and hence the cause-effect
structure of the considered IS. This is ultimately due to the fact
that we are not only talking of Shannon entropy of arbitrary random
variables or time series, but of goods exchanged from agents to patients
in a polyadic cause-effect relationships, see \cite[Sec.~3.4.1]{Gio24_IS}.
In fact, in every IS, the cause-effect link is expressed by activation
and occurrence times, see \cite[(CE)]{Gio24_IS} and \cite[(SA)]{Gio24_IS}.
Moreover, propagators are themselves interacting entities, and hence
goods can represent fluxes of any form between agents and patients,
see the long list of examples in \cite{Gio24_IS}, examples in Sec.~\ref{sec:Complex-adaptive-systems}
above, and Sec.~\ref{sec:Power-law-distribution} below, where goods
are probabilities, or Sec.~\ref{subsec:Maximization-of-diversification}
where goods are exchanged commodities. In other words, in order that
the GEP holds in a \emph{validated} IS we have to satisfy severe constraints,
because we have to respect the idea of fluxes of goods exchanged in
a polyadic relationships between agents and patients.

For example, the GEP can be used to distinguish the importance of
different messages/states (thought of as interacting entities, and
not as random variables) for a given population $\mathcal{P}$. After
all, the intuitive difference between a meaningful/readable message
and a completely random/unreadable one, is exactly that the former
is able to send us signals (propagators) that we are able to interpret
(goods), whereas random variables can be identified with their probability
distribution, which are more unintelligible for our brain.

In our opinion, this gives elements to solve the critiques expressed
by \cite{JaBaCr} in using transfer entropies in a mechanistic interpretation
as information flow.\smallskip{}

Similarly, it is also important to understand the differences between
GEP and the decreasing of entropy for out of equilibrium systems.
Indeed, it is well-known, see e.g.~\cite{Ni-Pr77}, that this type
of systems constantly make efforts to expel waste having a high level
of entropy, and this allows them to keep a \emph{lower} level of \emph{global}
entropy. Therefore, this seems to contradict the GEP, where the diversification
forces, related to the entropy of exchanged goods, have to be increased.
However, these wastes have a great level of entropy but, exactly because
they are expelled, do not have stable interactions relating their
agents with patients remaining in the system. This decreasing of entropy
is therefore completely different from the increasing of diversification:
in the latter, the system adapts if it strives \emph{to keep} these
exchanged goods and the information/diversification they represent.
\smallskip{}

Finally, the GEP is also very different from the second principle
of thermodynamics, because \emph{we are not considering the whole
entropy} of an isolated system but, on the contrary, only those related
to the exchanged goods of the interactions $i\in I_{\mathcal{P}}$
that aim at decreasing the average cost function $C$. This costs
essentially implies that the system is necessarily not isolated. Similar
considerations, essentially based on the remark that in the GEP we
do not calculate the entropy of the whole system, can be repeated
for nonadditive entropic functionals, see e.g.~\cite{Tsa23} and
references therein.

\section{Mathematical definition of the generalized evolution principle}

The previous motivations justify the following
\begin{defn}[Generalized evolution principle, CAS]
\label{def:CAS}Let $\mathfrak{I}$ be an IS, $s$, $t\in[\tst,\tend]$
and $\mathcal{P}\subseteq E$. Then, we say that \emph{in the sample
$\omega\in\Omega^{\text{g}}$ at time $t$ the population $\mathcal{P}$
is better adapted than at time $s$ (with respect to $I_{\mathcal{P}}$,
$\left(P_{y}^{j}\right)_{j,y}$, $\left(C_{y}^{j}\right)_{j,y}$,
$\left(Q_{\gamma}\right)_{\gamma}$, in the global space $\bar{M}_{\mathcal{P},I_{\mathcal{P}}}$;
briefly: $\mathcal{P}$ is a CAS) if}
\begin{enumerate}
\item $I_{\mathcal{P}}$ is a family of interactions of $\mathcal{P}$ called
\emph{adaptive interactions}.
\item $P_{y}^{j}$ is a probability on the population $\mathcal{P}$ for
each $j=1,\ldots,k$ and each $y\in\bar{M}_{\mathcal{P},I_{\mathcal{P}}}$.
The probability $P_{y}^{j}$ is called \emph{probability to average
the $j-$th cost at the state $y$}.
\item $C_{y}^{j}:\mathcal{P}\ra\R_{\ge0}$ is a measurable function with
respect to $P_{y}^{j}$ for each $y\in\bar{M}_{\mathcal{P},I_{\mathcal{P}}}$
and $j=1,\ldots,k$. $C_{y}=(C_{y}^{1},\ldots,C_{y}^{k})$ is called
\emph{unification cost function at $y$}.
\item $Q_{\gamma}$ is a probability on $I_{\mathcal{P}}$ for each $\gamma\in R_{\mathcal{P}}$
(see \eqref{eq:spacesGoods}) called \emph{diversification probability}.
\item We have
\begin{align}
U_{\mathcal{P}}^{j}(\omega,t) & \ge U_{\mathcal{P}}^{j}(\omega,s)\qquad\forall j=1,\ldots,k,\label{eq:unificationSample}\\
D_{I_{\mathcal{P}}}(\omega,t) & \ge D_{I_{\mathcal{P}}}(\omega,s),\label{eq:diversificationSample}
\end{align}
where \emph{unification forces $U_{\mathcal{P}}^{j}(\omega,t)$} are
defined by \eqref{eq:unificationForcesSample}, and\emph{ diversification
forces} $D_{I_{\mathcal{P}}}$ are defined by \eqref{eq:diversificationForces}.
\end{enumerate}
We also say that the state $\bar{y}=(x,\bar{t},N)\in\bar{M}_{\mathcal{P},I_{\mathcal{P}}}$
(see \eqref{eq:globalStatesAll}) is an \emph{emergent pattern for
$\mathcal{P}$} if
\begin{align}
U_{\mathcal{P}}^{j}(\bar{y}) & \ge U_{\mathcal{P}}^{j}(\omega,s)\qquad\forall j=1,\ldots,k,\label{eq:unificationState}\\
D_{I_{\mathcal{P}}}\left(\left(\gamma_{x,i}\right)_{i\in I_{\mathcal{P}}}\right) & \ge D_{I_{\mathcal{P}}}(\omega,s)\label{eq:diversificationState}
\end{align}
for all $\omega\in\Omega^{\text{g}}$ and all times $s\in[\tst,\tend]$.
We recall that $x\in\bar{M}_{\mathcal{P},t}^{\text{s}}$ and hence
$\left(\gamma_{x,i}\right)_{i\in I_{\mathcal{P}}}$ are the goods
included in the global state $\bar{y}=(x,\bar{t},N)\in\bar{M}_{\mathcal{P},I_{\mathcal{P}}}$.
\end{defn}

The latter conditions \eqref{eq:unificationState} and \eqref{eq:diversificationState}
can hold even if $\bar{y}$ is a global state which is not reached
by the system, i.e.~the equality $\bar{y}=Y_{\mathcal{P}}(\bar{\omega})|_{[\tst,t]}$
never holds for any $\bar{\omega}$ e any $t$. In fact, we may have
a step from a state $Y_{\mathcal{P}}(\omega)|_{[\tst,s]}$ to a better
state $Y_{\mathcal{P}}(\omega)|_{[\tst,t]}$ even if there does not
exist a best possible state, i.e.~an emergent pattern. An intuitive
example that seems to satisfy this property may come from Darwinian
evolution. Think at giraffes and their elongation of neck: the cost
are related to the probability of finding leaves to eat; we have at
least two interacting entities: giraffes and trees, but only one is
adapting with respect to the cost of being hungry; the population
interactions of the force of diversification allowed some giraffes
to have a genetic code that causes a longer neck. It seems that there
is not a maximum length of neck minimizing \emph{this} cost, even
if such a maximum is reached due to the increasing of \emph{other}
costs.

It is clear that a very general case, even if conceptually it cannot
exhaust all the possibilities, is when the cost function $C$ depends
only on the state $X_{\mathcal{P}}(t)$ and not on $(T_{\mathcal{P}}(t),N_{\mathcal{P}}(t))$.
This is an implicit assumption we will consider in the rest of the
paper.

If the cost function $C_{y}\equiv C_{y}(e)$ does not depend on the
entity $e\in\mathcal{P}$ nor on the global state $y$, and $\bar{y}$
is an emergent pattern, the GEP reduces to the classical entropy maximization
principle for the family of probabilities $\left(Q_{y}\right)_{y}$.
On the other hand, if we trivialize all the interactions $i\in I_{\mathcal{P}}$
by choosing constant resources (deterministic extraction of goods),
then the population is adapting at an emergent pattern $\bar{y}$
if it minimizes the expected value of the cost function $C_{\bar{y}}$.
Therefore, the GEP includes both the entropy maximization principle
and classical minimization problems of the cost function $y\in\bar{M}_{\mathcal{P},I_{\mathcal{P}}}\mapsto C_{y}^{j}$
if $C_{y}^{j}$ is deterministic with respect to the probability $P_{y}^{j}$.

Because the expected costs $E_{y}^{j}[C_{y}^{j}]$, $j=1,\ldots,k$,
in \eqref{eq:unificationForcesState} are calculated with a probability
distribution over the population $\mathcal{P}$, we can interpret
the unification forces $U_{\mathcal{P}}^{j}(y)=-E_{y}^{j}(C_{y}^{j})$,
$j=1,\ldots,k$, as proportional to a quantification of suitable \emph{common
goods for the population $\mathcal{P}$}. As a consequence of the
maximization properties of Shannon's entropy, the diversification
forces $D_{I_{\mathcal{P}}}(\gamma)$ can be interpreted as a gauge
of \emph{long-term changes} (with respect to the given cost function
$C$).

Note that an emergent patter $\bar{y}=(x,\bar{t},N)$ results into
probabilities $P_{\bar{y}}^{j}$ on the population $\mathcal{P}$
and $Q_{\bar{\gamma}}$ on adaptive interactions $I_{\mathcal{P}}$,
where $\bar{\gamma}=\left(\gamma_{x,i}\right)_{i\in I_{\mathcal{P}}}$
are the goods included in the global state $\bar{y}$. Starting from
these probability distributions, and depending on the problem, we
can then consider particular realizations $e\in\mathcal{P}$ and $i\in I_{\mathcal{P}}$,
like the mode (i.e.~where the probabilities assume highest value)
in frequently used AI algorithms, see e.g.~\cite{BaChBe}; in this
case, $e$ represents the interacting entity where the weight $P_{\bar{y}}^{j}(e)$
to evaluate the average unification cost $C_{\bar{y}}^{j}$ is highest,
whereas $i$ represent the adaptive interaction which contributes
more to the diversification $D_{I_{\mathcal{P}}}(\bar{\gamma})$.
We can clearly have more that one of these $e$ and $i$ in case of
multimodal probabilities.

Is it a classical complicated system, such as a spring-driven clock,
a CAS with respect to Def.~\ref{def:CAS}? We can think this clock
as a continuous dynamical system (maybe with suitable generalized
functions to represent the discrete ticking of the hands of the clock),
so that it is surely an IS, see \cite[Sec.~4.1]{Gio24_IS}. This IS
can be considered having simple cause-effect interactions starting
from the spring and arriving at clock's hands. The natural cost function
is hence the distance between the perfect speed of each hand, e.g.~$1\text{ Hz}$,
and its actual eventually lower speed. The natural main resource is
the potential energy stored in the spring. Therefore, the cost are
eventually increasing and the diversification forces always constant
(in a model with deterministic use of resources) or decreasing (stochastic
use of resources with a decreasing variance), and indeed also intuitively
the system is not adapting. On the other hand, note that the system
given by the spring clock + a winding person is a CAS.

Note that this mathematical formalization of the GEP has been possible
only thanks to the language we introduced for IS in \cite{Gio24_IS}:
interactions as cause-effect relation between agents and patients
and carried by propagators, state map and state space, space of resources
and goods; all these concepts are used in the previous definition.

We close this section by summarizing the steps we need to realize
in order to model a CAS satisfying the GEP.

\subsection{\label{subsec:General-steps-to-GEP}General steps to model a complex
adaptive system}

For the sake of clarity, we list here the general procedure we have
to follow in order to model a system which obeys the GEP. The process
presented here is clearly not linear and presents several cause-effect
polyadic feedback interactions, i.e.~it can be thought of as a kind
of meta-IS.
\begin{enumerate}[label=\arabic*.]
\item We clearly have to start by defining an IS, or another type of model
embedded as an IS. Therefore, in general we need (see \cite[Tables~1,~2,~3]{Gio24_IS}):
interacting entities, interactions, activations, proper state, goods
and resources, occurrence times, neighbourhoods and transition functions,
even if in several particular models some of these notions are trivial.
\item Understand what are good global states $\bar{M}_{\mathcal{P},t}^{\text{s}}$,
$M_{J,t}^{\text{t}}$, $M_{J,t}^{\text{n}}$, see Sec.~\ref{sec:Power-law-distribution}
and Sec.~\ref{subsec:Generalizing-Von-Thunen} for examples and the
next point in this list.
\item We can try to define the global probability space $\Omega^{\text{g}}$
as explained in Sec.~\ref{sec:Global-state} above. However, an equivalent
approach is to understand only the probability distributions of the
processes $X$, $T$, $N$ of \eqref{eq:randVarIS} and set $\Omega^{\text{g}}:=\bar{M}_{\mathcal{P},I_{\mathcal{P}}}$
with probability $P^{\text{g}}$ given by the joint probabilities
of $X$, $T$, $N$. These processes $X$, $T$, $N$ are then represented
by the projections of each $\bar{M}_{\mathcal{P},I_{\mathcal{P}},t}$
resp.~on $\bar{M}_{\mathcal{P},t}^{\text{s}}$, $\bar{M}_{I_{\mathcal{P}},t}^{\text{t}}$
and $\bar{M}_{I_{\mathcal{P}},t}^{\text{n}}$.
\item Understand what are the unification costs $\left(C_{y}^{j}\right)_{j,y}$
causing unification forces. These are what the system tries to avoid
as far as possible using adapting interactions.
\item Define probabilities to average unification costs $\left(P_{y}^{j}\right)_{j,y}$,
e.g.~giving a greater weight at global states and entities intuitively
paying greater costs.
\item Define diversification probabilities $\left(Q_{\gamma}\right)_{\gamma}$;
these are related to the goods exchanged during the dynamics in order
to distribute resources and decrease the expected value of costs.
Therefore, $Q_{\gamma}(i)$ will be higher for important adapting
interactions $i\in I_{\mathcal{P}}$.
\item Understand what are the\emph{ }adaptive interactions $I_{\mathcal{P}}$,
i.e.~the interactions of the system that dynamically improve the
state of the system by following the GEP.
\item Assume to have an emergent pattern $\bar{y}\in\bar{M}_{\mathcal{P},I_{\mathcal{P}}}$
and derive necessary conditions, or consider simulations where unification
$U_{\mathcal{P}}^{j}(\omega,t)$ and diversification $D_{I_{\mathcal{P}}}(\omega,t)$
increases, see Sec.~\ref{sec:Power-law-distribution} and Sec.~\ref{subsec:Generalizing-Von-Thunen}
for examples.
\item Validate the model by comparing, in the strongest possible way, emergent
patterns predicted from your model with configurations of systems
that are clearly independent from the model itself, e.g.~real-world
systems.
\end{enumerate}
\LyXZeroWidthSpace{}

We now prove, in a more abstract setting and under mild conditions,
that in every CAS the diversification probabilities $Q_{\gamma}$
satisfy a power law when the population is at an emergent pattern,
i.e.~when unification and diversification are maximum. We then give
some examples of cost functions in Sec.~\ref{subsec:ExamplesCostsFunct}.

\section{\label{sec:Power-law-distribution}Power law distribution following
Mandelbrot's idea}

It is well known that power law distributions are frequently associated
to CAS, and hence appear often both in nature and in social systems,
see e.g.~\cite{Sor06,Gab99} and references therein.

Using the language of IS theory, this section follows and generalizes
the classical ideas of B.~Mandelbrot presented in \cite{Man53}.

We imagine to have an IS and a population $\mathcal{P}$ whose state
is described by vectors $x\in S_{\mathcal{P}}\subseteq\R^{n}$. The
systems has to be thought of as a CAS that changes its state so as
to decrease a suitable cost function $E_{\mathcal{P}}:S_{\mathcal{P}}\ra\R_{>0}$
and, at the same time, to increase a corresponding information entropy:
\begin{equation}
D_{I_{\mathcal{P}}}(x)=-\sum_{j=1}^{d}q_{j}(x_{j})\cdot\log_{2}q_{j}(x_{j})>0\qquad\forall x\in S_{\mathcal{P}}.\label{eq:diversMand}
\end{equation}
For example, as in \cite{Man53}, we can think at $n$ words of an
evolving language whose state is described by frequencies $q_{j}(x_{j})=x_{j}$
for each word $j=1,\ldots,n$. Using the language of IS theory, we
can think to have interacting entities $E=\{1,\ldots,n\}=\mathcal{P}$
and $d\le n$ interactions of the type $i_{j}:j\xra{j,\text{evo}}j$
for each $j\in E$ representing the evolution of the $j-$th interacting
entities from a previous (unknown) state to the new one $x_{j}$.
Therefore, $x_{j}=\gamma_{i_{j}}(t)$ are the goods of $i_{j}$, and
all the other elements of the IS are trivial, e.g.~occurrence times
and transition functions $f_{j}$ are trivial. The adapting interactions
are hence $I_{\mathcal{P}}=\{i_{j_{1}},\ldots,i_{j_{d}}\}$, and the
diversification probabilities are $Q_{x}(i_{j_{k}})=q_{j_{k}}(x_{j_{k}})$
for all $k=1,\ldots,d$ (compare \eqref{eq:diversMand} and \eqref{eq:DiversificationFinite}).
We are only interested to the global state space $\bar{M}_{\mathcal{P},I_{\mathcal{P}}}=S_{\mathcal{P}}\subseteq\R^{n}$
because all the other ones are trivial (no activation, times, neighbourhoods,
etc.).

In general, i.e.~not just thinking at the example of an evolving
language, we assume that the probabilities always enter into the global
state of the goods of the population $\mathcal{P}$:
\[
q_{j}(x_{j})=x_{j}\qquad\forall j=1,\ldots,d,\ \forall x\in S_{\mathcal{P}}.
\]
The system adapts (e.g.~it evolves) so as to minimize the ratio $\frac{E_{\mathcal{P}}}{D_{I_{\mathcal{P}}}}$
at the interior point $y\in S_{\mathcal{P}}$:
\[
\forall x\in S_{\mathcal{P}}:\ 0<\frac{E_{\mathcal{P}}(y)}{D_{I_{\mathcal{P}}}(y)}\le\frac{E_{\mathcal{P}}(x)}{D_{I_{\mathcal{P}}}(x)}.
\]
Note that this condition is surely satisfied if the system follows
the GEP and $E_{\mathcal{P}}(x)$ is the average value of a cost function
$C_{x}$ defined on the population $\mathcal{P}$, see below. 

\noindent The cost function $E_{\mathcal{P}}$ must satisfy the inequality
\begin{equation}
\frac{\partial E_{\mathcal{P}}}{\partial x_{k}}(y)\le\alpha_{k}(y)\cdot\log_{2}k\qquad\forall k=2,\ldots,d\label{eq:ineqParDerCost}
\end{equation}
for some $\alpha_{k}:S_{\mathcal{P}}\ra\R$. Note that \eqref{eq:ineqParDerCost}
is not required to hold for $k=1$. Finally, we assume that
\begin{equation}
\sum_{k=1}^{d}k^{-\alpha_{k}(y)\cdot\frac{D_{I_{\mathcal{P}}}(y)}{E_{\mathcal{P}}(y)}}=:N(y)\ge\frac{1}{q_{1}(y)}\ge e.\label{eq:normalHp}
\end{equation}
Note that this implies $q_{1}(y)\le e^{-1}\simeq0.368$. This is another
restriction on the value $\alpha_{k}(y)\in\R$: whereas condition
\eqref{eq:ineqParDerCost} states that we can take $\alpha_{k}(y)$
as large as we want, the inequality \eqref{eq:normalHp} yields that
the larger is $\alpha_{k}(y)$, the more difficult will be to arrive
at a value $N(y)\ge e$ because of the dependence of the normalization
factor $N(y)$ from $\alpha_{k}(y)$.

We have the following
\begin{thm}
\label{thm:powerLaw}Let $S_{\mathcal{P}}\subseteq\R^{n}$ be an open
set and $y\in S_{\mathcal{P}}$. Let $q_{j}\in\mathcal{C}^{1}(S_{\mathcal{P}},\R_{\ge0})$,
for all $j=1,\ldots,d\le n$, be such that
\[
\forall x\in S_{\mathcal{P}}:\ \left(q_{j}(x)\right)_{j=1,\ldots,d}\text{ is a probability.}
\]
Let
\[
D_{I_{\mathcal{P}}}(x)=-\sum_{j=1}^{d}q_{j}(x_{j})\cdot\log_{2}q_{j}(x_{j})\qquad\forall x\in S_{\mathcal{P}}
\]
be the diversification force $D_{I_{\mathcal{P}}}$. Let $E_{\mathcal{P}}\in\mathcal{C}^{1}(S_{\mathcal{P}},\R_{>0})$
be such that
\begin{equation}
\forall x\in S_{\mathcal{P}}:\ 0<\frac{E_{\mathcal{P}}(y)}{D_{I_{\mathcal{P}}}(y)}\le\frac{E_{\mathcal{P}}(x)}{D_{I_{\mathcal{P}}}(x)}.\label{eq:min}
\end{equation}
For example, this holds if $y\in S_{\mathcal{P}}$ is an emergent
pattern for $\mathcal{P}$. Finally, assume that
\begin{align}
 & q_{j}(x_{j})=x_{j}\qquad\forall j=1,\ldots,d\ \forall x\in S_{\mathcal{P}}\nonumber \\
 & \frac{\partial E_{\mathcal{P}}}{\partial x_{k}}(y)\le\alpha_{k}(y)\cdot\log_{2}k\qquad\forall k=2,\ldots,d\label{eq:hpParC}\\
 & \sum_{k=1}^{d}k^{-\alpha_{k}(y)\cdot\frac{D_{I_{\mathcal{P}}}(y)}{E_{\mathcal{P}}(y)}}=:N(y)\ge\frac{1}{q_{1}(y)}\ge e,\nonumber 
\end{align}
where $\alpha_{k}:S_{\mathcal{P}}\ra\R$. Then, we have
\begin{enumerate}
\item $q_{k}(y)=q_{1}(y)\cdot k^{-\alpha_{k}(y)\cdot\frac{D_{I_{\mathcal{P}}}(y)}{E_{\mathcal{P}}(y)}}$
for all $k=1,\ldots,d$.
\item $q_{1}(y)=\frac{1}{N(y)}$.
\end{enumerate}
\end{thm}

\begin{proof}
Since $S_{\mathcal{P}}$ is an open set, $y\in S_{\mathcal{P}}$ and
$E_{\mathcal{P}}$, $D_{I_{\mathcal{P}}}\in\mathcal{C}^{1}(S_{\mathcal{P}},\R_{>0})$,
from \eqref{eq:min} we get $\partial_{k}\left(\frac{E_{\mathcal{P}}}{D_{I_{\mathcal{P}}}}\right)(x)=0$,
where $\partial_{k}:=\frac{\partial}{\partial x_{k}}$. For simplicity,
all the functions that will appear in the following are evaluated
at the point $y$. We have
\[
\partial_{k}E_{\mathcal{P}}\cdot D_{I_{\mathcal{P}}}+E_{\mathcal{P}}\cdot\sum_{j}\left(\partial_{k}y_{j}\cdot\log_{2}y_{j}+y_{j}\frac{1}{y_{j}}\log_{2}e\cdot\partial_{k}y_{j}\right)=0,
\]
Where we used $q_{j}(y)=y_{j}$ and hence $\partial_{k}y_{j}=\frac{\partial}{\partial x_{k}}(x_{j})|_{x_{j}=y_{j}}=\delta_{kj}$,
so
\[
D_{I_{\mathcal{P}}}\cdot\partial_{k}E_{\mathcal{P}}+E_{\mathcal{P}}\left(\log_{2}y_{k}+\log_{2}e\right)=0.
\]
By \eqref{eq:hpParC}, for $k\ge2$ we obtain
\[
\log_{2}y_{k}=-\frac{D_{I_{\mathcal{P}}}}{E_{\mathcal{P}}}\partial_{k}E_{\mathcal{P}}-\log_{2}e\ge-\frac{D_{I_{\mathcal{P}}}}{E_{\mathcal{P}}}\alpha_{k}\log_{2}k-\log_{2}e,
\]
and hence
\begin{equation}
y_{k}=q_{k}\ge2^{\left(-\frac{D_{I_{\mathcal{P}}}}{E_{\mathcal{P}}}\alpha_{k}\log_{2}k-\log_{2}e\right)}=k^{-\alpha_{k}\frac{D_{I_{\mathcal{P}}}}{E_{\mathcal{P}}}}e^{-1}.\label{eq:base2}
\end{equation}
We assumed that $N\ge e$, so
\[
q_{k}\ge\frac{1}{N}\cdot k^{-\alpha_{k}\frac{D_{I_{\mathcal{P}}}}{E_{\mathcal{P}}}}\qquad\forall k=1,\ldots,d.
\]
Note that this inequality holds also for $k=1$ because we assumed
that $q_{1}(y)\ge\frac{1}{N(y)}$.

Finally, we note that we cannot have $q_{h}>\frac{1}{N}\cdot h^{-\alpha_{k}\frac{D_{I_{\mathcal{P}}}}{E_{\mathcal{P}}}}$
for some $h=1,\ldots,d$ because otherwise we would have
\[
\sum_{j=1}^{d}q_{j}=1>\sum_{j=1}^{d}\frac{1}{N}k^{-\alpha_{k}\frac{D_{I_{\mathcal{P}}}}{E_{\mathcal{P}}}}=1.
\]
Therefore, we must have $q_{k}=\frac{1}{N}k^{-\alpha_{k}\frac{D_{I_{\mathcal{P}}}}{E_{\mathcal{P}}}}$
for all $k=1,\ldots,d$. Finally, for $k=1$ we get $q_{1}=\frac{1}{N}$
which proves both our conclusions.
\end{proof}
\noindent Examples of cost functions satisfying \eqref{eq:ineqParDerCost}
are given by the average value $E_{\mathcal{P}}(x)=\sum_{k=1}^{n}c_{k}(x)\cdot p_{k}(x)$,
where $\left(p_{1}(x),\ldots,p_{n}(x)\right)$ is a probability. In
the notations of the GEP, we have $U_{\mathcal{P}}(x)=-E_{\mathcal{P}}(x)$,
the unification cost is $C_{x}(k)=c_{k}(x)$, and the probabilities
to average the cost are $P_{x}(k)=p_{k}(x)$ for each interacting
entity $k\in\mathcal{P}=\{1,\ldots,n\}$.

The $k$-th component of the cost $c_{k}$ can be any one of the following
examples.

\subsection{\label{subsec:ExamplesCostsFunct}Examples of cost functions}
\begin{enumerate}
\item \label{enu:exInver}In this first case, $p_{k}(x)=q_{k}(x)=x_{k}$
and $c_{k}(x)=\frac{a}{q_{k}(x)^{s}}\log_{2}k=\frac{a}{x_{k}^{s}}\log_{2}k$,
where $a$, $s\in\R_{>0}$. This example represents costs that are
decreasing with an increasing of the probabilities $q_{k}(x)=x_{k}$
and they are increasing with an increasing of the number of bits $\log_{2}k$
which are necessary to transmit the rank $k$. In this case we have
$\partial_{k}E_{\mathcal{P}}(y)=\partial_{k}c_{k}(y)\cdot y_{k}+c_{k}(y)=\frac{a}{y_{k}^{s}}\left(1-s\right)\log_{2}k$.
Therefore, it suffices to take $\alpha_{k}(y)=\frac{a}{y_{k}^{s}}\left(1-s\right)$.
\item For $k\ge2$, assume that $c_{k}(x)\le a\cdot\log_{b}(k+k_{0})+j_{0}$,
where $a$, $j_{0}\in\R_{>0}$, $b>2$, $k_{0}\le b-2$ and also that
$\partial_{k}E_{\mathcal{P}}(x)\le c_{k}(x)$. This example is considered,
with the equality signs, in \cite{Man53}. Usually, $b$ is the number
of letters in an alphabet. The second inequality e.g.~holds if the
costs $c_{k}$ do not depend on the probabilities $y_{k}$, so that
$\partial_{k}E_{\mathcal{P}}(x)=c_{k}(x)$. We have $\log_{b}(k+k_{0})\le\log_{2}k$
if and only if
\begin{equation}
k+k_{0}\le b^{\log_{2}k}=\left(b^{\log_{b}k}\right)^{\frac{1}{\log_{b}2}}=k^{\log_{2}b}.\label{eq:ineq_k_0}
\end{equation}
Since $b>2$, we have $\log_{2}b>1$ and the function $k^{\log_{2}b}-k$
is increasing in $k$. Therefore if $k_{0}\le2^{\log_{2}b}-2=b-2$
the inequality \eqref{eq:ineq_k_0} always holds, and hence $c_{k}(x)\le a\log_{2}k+j_{0}$.
If we take $\alpha_{k}(x)\equiv a+j_{0}$, we get $\alpha_{k}\log_{2}k=a\log_{2}k+j_{0}\log_{2}k\ge a\log_{2}k+j_{0}\ge a\cdot\log_{b}(k+k_{0})+j_{0}\ge c_{ik}(x)\ge\partial_{k}E_{\mathcal{P}}(x)$
for $k\ge2$.
\item We can choose a better estimate of $\alpha_{k}$ if, for $k\ge2$,
$c_{k}(y)\le j_{0}+a\cdot\log_{b}k$ and $\partial_{k}E_{\mathcal{P}}(x)\le c_{k}(x)$,
i.e.~if $k_{0}=0$ in the previous example. This case is considered
in \cite{CoMi}. In fact, for $k\ge2$ we have
\begin{align*}
\partial_{k}E_{\mathcal{P}}(y) & \le c_{k}(y)\le j_{0}+a\log_{b}k=j_{0}\log_{2}2+a\log_{b}2\cdot\log_{2}k\le\\
 & \le j_{0}\log_{2}k+a\log_{b}2\cdot\log_{2}k=(a\log_{b}2+j_{0})\cdot\log_{2}k.
\end{align*}
We can hence set $\alpha_{k}\equiv a\log_{b}2+j_{0}$.
\item $c_{k}(y)=\gamma_{k}>0$. This is the case of constant costs depending
only on the rank $k$ (e.g.~we can have that $\gamma_{k}$ is proportional
to the length of the words of rank $k$). We therefore have to take
$\alpha$ so that $\alpha_{k}(y)\ge\frac{\gamma_{k}}{\log_{2}k}$.
\end{enumerate}
\noindent Note that in the first example \ref{enu:exInver} it is
not reasonable to assume that this formula holds also for $k=1$ because
this would yield a null cost $c_{1}(y)=0$. This, and the calculations
we realized in the second example, motivate that \eqref{eq:ineqParDerCost}
holds only for $k\ge2$.

\section{\label{subsec:Generalizing-Von-Thunen}Generalizing Von Thünen's
model}

Von Thünen's model, see \cite{VonTh}, tries to answer the basic questions
of location and land use theory: ``where should a certain activity
be located?'' and ``which activity should be chosen at a certain
location?''. Both questions address the principles underlying the
spatial layout of an economy. Several original assumptions of Von
Thünen's model can be weakened and generalized by simply using an
appropriate mathematical notations and modeling; later, we will see
them in details. Most important for us is that, in this model, costs
and forces of diversification have simple properties that can be generalized
to other CAS. This actually show another feature of having a common
mathematical language for CS: describing a particular model with the
language of IS theory, one can recognize that the obtained results
can be actually generalized and hence potentially applied to several
other CS.

In the present section, we explicitly use units of measurement, because
this greatly helps the understanding of the different economic quantities
that we are going to introduce. We will use notations such as $\left[\frac{\text{€}}{\text{m}^{2}}\right]$
to denote the 1-dimensional (totally ordered real vector) space of
quantities whose unit of measurement is € per square meter. Mathematically,
it can be identified with the space of polynomials in the unknowns
$\text{€\ensuremath{\cdot\text{m}}}^{-2}$.

\subsection{\label{subsec:VT-impedanceZones}Von Thünen's impedance zones}

We need to introduce several quantities and notations:
\begin{enumerate}[label=\arabic*.]
\item $B\in\N_{>0}$ the \emph{number of commodities} produced by the considered
economy. There are no a priory assumptions on the type of commodities
(e.g.~not necessarily of agricultural type).
\item For all $b=1,\ldots,B$, a unit of measurement $u_{b}$ for the commodity
$b$. For example $u_{b}$ can be $\textsf{ton}$, or an integer number
$n\in\N_{>0}$ so that $[u_{b}]=\R$ (dimensionless), or $\textsf{box}$,
etc. 
\item $x_{m}\in\R^{2}$ is the \emph{location of the market} $m$.
\item $A\subseteq\R^{2}$ all possible \emph{locations for the companies}
that produce some commodity $b=1,\ldots,B$.
\item \label{enu:yield}$y_{b}:A\ra\left[\frac{u_{b}}{\text{m}^{2}}\right]_{>0}$,
$y_{b}(x)>0$ is the \emph{yield of the commodity} $b$ if the company
is located at $x\in A$. The model is not stochastic, so quantities
like these can denote average values. The use of a spatial unit of
measurement like $\text{m}^{2}$ is only useful so as to not use too
heavy notations. More appropriate units of production $v_{b}$, depending
on the commodity $b$, could be introduced instead of $\text{m}^{2}$
(e.g.~it could be $v_{b}=\text{kwh}$, or $v_{b}=\frac{\text{hour}}{\text{man}}$). 
\item $p_{b}(x_{m})\in\left[\frac{\text{€}}{u_{b}}\right]_{>0}$ \emph{price
of the commodity} $b$ in the market $x_{m}$ per units of $b$.
\item $c_{b}:A\ra\left[\frac{\text{€}}{u_{b}}\right]_{>0}$ \emph{, $c_{b}(x)>0$
}is the \emph{production cost} of $b$ at the location $x\in A$ per
units of $b$.
\item $j(-,x_{m}):A\ra[i]_{\ge0}$, $j(x,x_{m})$ is the \emph{impedance}
between the location $x\in A\subseteq\R^{2}$ and the market $x_{m}\in\R^{2}$
measuring the pure transportation costs to move from $x$ to $x_{m}$
(see also below); the impedance has unit of measurement $i$ (e.g.~$i$
can be $i=\text{km}$, $i=\text{hour}$, $i=\frac{\text{hour}}{\text{man}}$,
$i=\text{€}$, etc.). We do not need any assumption about existence
or non existence of possible routes, neither on the nature of the
transportation, nor on the linearity of $j$ with respect to the distance
between $x$ and $x_{m}$ along the shortest route, not even if we
are actually moving goods from $x$ to $x_{m}$ or vice versa (this
is only a useful intuitive interpretation in case of application in
land use theory).
\item $F_{b}:[i]_{\ge0}\ra\left[\frac{\text{€}}{u_{b}}\right]_{>0}$,
$F_{b}(d)$ is the \emph{transportation cost} of the commodity $b$
per units of $b$ and for any pair of points, $(x,x_{m})\in\R^{2}$
having impedance $d\in[i]_{>0}$. For example, $F_{b}(d)$ can be
lower if we assume the possibility to use refrigerators for the transportation
of a dairy product $b$. This modeling assumption includes the particular
case where $F_{b}(d)=\bar{F}_{b}\cdot d$, where $\bar{F}_{b}\in\left[\frac{\text{€}}{u_{b}\cdot i}\right]_{>0}$,
i.e.~the case where the transportation cost is proportional to the
impedance $d$. Usually, one assumes that the transportation cost
is increasing with the impedance $d$:
\begin{equation}
\forall d_{1},d_{2}\in[i]_{\ge0}:\ d_{1}\le d_{2}\Rightarrow F_{b}(d_{1})\le F_{b}(d_{2}).\label{eq:F_b-increasing}
\end{equation}
However, we will never use this assumption.
\item $k_{b}(-,x_{m}):A\ra\left[\frac{\text{€}}{\text{m}^{2}}\right]_{>0}$,
$k_{b}(x,x_{m})$ is the average \emph{life cost} (everything but
the rent of the company's location), per $\text{m}^{2}$, of the owner
producing $b$ at $x$ and selling $b$ in the market $x_{m}$.
\end{enumerate}
We can now introduce the basic quantities of Von Thünen's model, using
both a specific language of land use theory and a more general one:
\begin{defn}
\noindent \label{def:value}The \emph{(land) value }or \emph{locational
rent} at $x\in A$ with respect to the production of $b=1,\ldots,B$
and the market $x_{m}$ is
\begin{equation}
L_{b}(x,x_{m}):=y_{b}(x)\cdot\left[p_{b}(x_{m})-c_{b}(x)-F_{b}\left(j(x,x_{m})\right)\right]\in\left[\frac{\text{€}}{\text{m}^{2}}\right].\label{eq:landValue}
\end{equation}
The \emph{ideal value/rent }for a company located \emph{at $x$ with
respect to the market located at $x_{m}$} is
\begin{equation}
R(x,x_{m}):=\max_{b=1,\ldots,B}\left[L_{b}(x,x_{m})-k_{b}(x,x_{m})\right]\in\left[\frac{\text{€}}{\text{m}^{2}}\right].\label{eq:idealRent}
\end{equation}
From this, we deduce that the life cost must satisfy the constraint
\begin{equation}
k_{b}(x,x_{m})<L_{b}(x,x_{m})\quad\forall x\in A.\label{eq:cost-value}
\end{equation}
Moreover, we say that $x$ \emph{is a good (location) for (the production
of) $b$} if
\begin{equation}
R(x,x_{m})=L_{b}(x,x_{m})-k_{b}(x,x_{m}),\label{eq:goodFor_b}
\end{equation}
i.e.~if the ideal rent $R$ equals the land value $L_{b}$ minus
the life cost $k_{b}$. Finally, the \emph{impedance boundaries around
the sink/market $x_{m}$} are given by
\begin{align}
\underline{r}_{b}(x_{m}): & =\inf\left\{ j(x,x_{m})\mid x\in A,\ x\text{ is good for }b\right\} \in[i]_{\ge0}\label{eq:inf}\\
\overline{r}_{b}(x_{m}): & =\sup\left\{ j(x,x_{m})\mid x\in A,\ x\text{ is good for }b\right\} \in[i]_{\ge0}.\label{eq:sup}
\end{align}
\end{defn}

\noindent Therefore, if a company is located at $x$, which is a good
location for the production of $b$, we trivially have that $\underline{r}_{b}(x_{m})\le j(x,x_{m})\le\overline{r}_{b}(x_{m})$,
that is the company is located in the corresponding \emph{impedance
zone} bounded by $\underline{r}_{b}(x_{m})\le\overline{r}_{b}(x_{m})$.
Note that if at least two locations are good for $b$ and have different
impedance, the zone is non trivial, i.e. $\underline{r}_{b}(x_{m})<\overline{r}_{b}(x_{m})$.

In this setting, only definition \eqref{eq:landValue} is really specific
of this model of land use. For an arbitrary IS, we can assume to have
a \emph{value function} $L_{b}:A\times A\ra\R$, a \emph{cost function}
$k_{b}:A\times A\ra\R$ satisfying \eqref{eq:cost-value} and an \emph{impedance
function} $j(-,x_{m}):A\ra[i]_{\ge0}$. With these, we can define
the quantities $R(x,x_{m})$ as in \eqref{eq:idealRent} (the \emph{ideal
value} for all the indexes $b=1,\ldots,B$), the property of being
a \emph{good} $x$ \emph{for $b$}, as in \eqref{eq:goodFor_b} and
the impedance boundaries as in \eqref{eq:inf} and \eqref{eq:sup}.

\subsection{Disjoint impedance zones}

We now want to see under what assumptions the impedance zones are
disjoint, i.e.~when $\underline{r}_{\beta}(x_{m})\le\overline{r}_{\beta}(x_{m})\le\underline{r}_{b}(x_{m})\le\overline{r}_{b}(x_{m})$
if $b\ne\beta$ are two different indexes/commodities. We need to
hypothesize that the cost of life $k_{b}$ does not depend on the
location $x$. We will use e.g.~the notation $k_{b}\equiv k_{b}(-)\in\left[\frac{\text{€}}{\text{m}^{2}}\right]$.
In land use theory, this is clearly an assumption which holds only
if $A\subseteq\R^{2}$ is not too large; e.g.~it surely does not
hold for locations situated in different countries, with different
cost of labor, different climate conditions and different life costs. 

Let $b$, $\beta=1,\ldots,B$, $b\ne\beta$, be two commodities. For
simplicity, we omit the dependence by the market's location $x_{m}$.
By contradiction, assume that
\begin{equation}
\overline{r}_{\beta}>\underline{r}_{b}\label{eq:star}
\end{equation}
Definition \eqref{eq:sup} yields the existence of a location $x\in A$
which is good for $\beta$ and such that $\underline{r}_{b}<j(x)\le\overline{r}_{\beta}$.
Since $x$ is good for $\beta$, we have
\begin{equation}
R(x)=L_{\beta}(x)-k_{\beta}=\max_{b}\left[L_{b}(x)-k_{b}\right]\ge L_{b}(x)-k_{b}.\label{eq:1_14}
\end{equation}
Analogously, from \eqref{eq:inf} and $\underline{r}_{b}<j(x)$ we
get the existence of $y\in A$ which is good for $b$ and such that
\begin{align}
\underline{r}_{b} & \le j(y)<j(x)\label{eq:3_14}\\
R(y) & =L_{b}(y)-k_{b}\ge L_{\beta}(y)-k_{\beta}.\label{eq:2_14}
\end{align}
How can it happen that $x$ is a good location for $\beta$ and $y$
is not a good location for it, even if $j(y)<j(x)$? To understand
this point, we compare the \emph{gain (of net land value) passing
from $x$ to }$y$, for an arbitrary commodity $\beta$ we have:
\begin{equation}
\Delta\eta_{\beta}(x,y):=\left[L_{\beta}(y)-k_{\beta}\right]-\left[L_{\beta}(x)-k_{\beta}\right]=L_{\beta}(y)-L_{\beta}(x)\label{eq:1_15}
\end{equation}
If we assume $\Delta\eta_{\beta}(x,y)\ge\Delta\eta_{b}(x,y)$ and
consider \eqref{eq:1_15}, we get
\begin{align*}
L_{\beta}(y) & \eqUp^{\eqref{eq:1_15}}L_{\beta}(x)+\Delta\eta_{\beta}(x,y)\ge L_{\beta}(x)+\Delta\eta_{b}(x,y)\\
 & \eqUp^{\eqref{eq:1_15}}L_{\beta}(x)+L_{b}(y)-L_{b}(x)\\
 & \phantom{(}\ge L_{\beta}(x)+L_{b}(y)-L_{b}(x)-k_{\beta}+k_{\beta}\\
 & \geUp^{\eqref{eq:1_14}}L_{b}(x)-k_{b}+L_{b}(y)-L_{b}(x)+k_{\beta}\\
 & \phantom{(}=-k_{b}+L_{b}(y)+k_{\beta}
\end{align*}
and hence $L_{\beta}(y)-k_{\beta}\ge L_{b}(y)-k_{b}$ so that $L_{\beta}(y)-k_{\beta}=L_{b}(y)-k_{b}$
by \eqref{eq:2_14}. Therefore, $y$ is a good location both for $\beta$
and $b$. We therefore proved that if $\Delta\eta_{\beta}(x,y)\ge\Delta\eta_{b}(x,y)$
and the impedance zones of $b$ and $\beta$ intersect, then $L_{\beta}(y)-k_{\beta}=L_{b}(y)-k_{b}$.

We can summarize this arguments with the following
\begin{thm}[J.H.~von Thünen]
Let us assume that $k_{b}$ does not depend on the location $x\in A$.
Let $b$, $\beta$ be two commodities such that for all $x$, $y\in A$
\begin{align}
\Delta\eta_{\beta}(x,y) & \ge\Delta\eta_{b}(x,y)\label{eq:Tr}\\
L_{\beta}(y)-k_{\beta} & \ne L_{b}(y)-k_{b}.\label{eq:LV}
\end{align}
Then impedance zones of $b$ and $\beta$ (around the market $x_{m}$)
are disjoint, i.e.~$\underline{r}_{\beta}\le\overline{r}_{\beta}\le\underline{r}_{b}\le\overline{r}_{b}$.
Therefore, if any pair of different commodities always have different
land values \eqref{eq:LV} and different variations of transportation
costs \eqref{eq:Tr} (for all $x$, $y\in A$), then we can order
the commodities so that
\[
\underline{r}_{b_{1}}\le\overline{r}_{b_{1}}\le\underline{r}_{b_{2}}\le\overline{r}_{b_{2}}\le\ldots\le\underline{r}_{b_{B}}\le\overline{r}_{b_{B}}.
\]
The same result holds in any IS where we can define a value function
$L_{b}:A\ra\R$, a cost $k_{b}\in\R$ and an impedance function $j:A\ra[i]_{\ge0}$
for each interacting entities $b\in\mathcal{P}$ in a finite population,
and for the quantities as defined in Def.~\ref{def:value}. Note
that the moving of commodities from $x\in A$ to the market $x_{m}$
is only an intuitive interpretation in land use theory, but in more
general IS we can also have the opposite movement.
\end{thm}

\noindent All this serves to underscores that we can have disjoint
impedance zones even if we do \emph{not} have an optimized economy,
i.e.~a complex adaptive economy. Therefore, this is not a peculiarity
of a CAS because we are not minimizing any cost nor maximizing any
diversification force. On the contrary, in the next section, we will
consider what happen when the subpopulation of companies producing
the same commodity adapts following the GEP, i.e.~evolves decreasing
natural costs in the most diversified way.

\subsection{Von Thünen's model and the generalized evolution principle}

In this section, we assume that our IS satisfies the GEP, i.e.~it
is a CAS, and it evolved into a stationary emergent pattern configuration,
where natural costs are minimum and diversification forces are maximum.
For simplicity, we still continue to use a language of land use theory,
for example assuming to have $n_{b}\in\N_{>0}$ companies producing
the commodity $b=1,\ldots,B$. However, it is clear that exactly the
same deductions apply to any IS where the interacting entities are
$E=\{(a,b)\mid b=1,\ldots,B,\ a=1,\ldots,n_{b}\}=\mathcal{P}$ and
we have a value function $L_{b}:A\ra\R$ and a cost function $k_{b}:A\ra\R$
for all $b=1,\ldots,B$. We do not actually even need an impedance
function $j:A\ra[i]_{\ge0}$ as above.

\subsubsection{Cost minimization}

In a non adapted economy, we do not necessarily have that rents coincide
with their ideal values \eqref{eq:idealRent}, e.g.~because of ignorance
of some agent with respect to the entire market configuration. For
simplicity, in this section we assume to consider only one market
$x_{m}$, so that all the prices, rents and values always refer to
$x_{m}$ and we will omit this variable. Then, the configuration space
of a generic economy located around the market $x_{m}$ is given by:
\begin{enumerate}
\item A location $x_{b}^{a}\in A$ for each company $a=1,\ldots,n_{b}$
producing the commodity $b=1,\ldots,B$. We assume that $x_{b}^{a_{1}}\ne x_{b}^{a_{2}}$
if $a_{1}\ne a_{2}$.
\item In each one of these locations, we have a rent $r_{b}^{a}\in\left[\frac{\text{€}}{\text{m}^{2}}\right]$
really applied to the company $a$ producing $b$ at $x_{b}^{a}$.
\end{enumerate}
\noindent Using these notations, we can model configurations representing
a positive cost for some agent acting as location renter or tenant
owning a company. For example the inequalities
\[
r_{b}^{a}<L_{b}(x_{b}^{i})-k_{b}(x_{b}^{a})<R(x_{b}^{a})
\]
imply
\begin{enumerate}
\item $r_{b}^{a}<L_{b}(x_{b}^{a})-k_{b}(x_{b}^{a})$: the location $x_{b}^{a}$
is rented at a price $r_{b}^{a}$ which is less than the land value
(it could be rented at a higher price).
\item $r_{b}^{a}<R(x_{b}^{a})$: the rent $r_{b}^{a}$ is less than the
maximum rent that it would be possible to ask in the location $x_{b}^{a}$
(at another company producing a different commodity).
\end{enumerate}
Considering $<$, $=$ or $>$, all the possible inequalities are
$3^{3}=27$, but several of them are mathematically impossible, repetitions,
or impossible relations due to the definition of $R(x)$ (see \eqref{eq:idealRent}).
Only the following possible inequalities remain:
\begin{defn}
\label{def:confSpaceEconomy}The configuration space $M=M_{\mathcal{P},I_{\mathcal{P}}}$
of the economies centered around the market $x_{m}$ is defined by
\[
(x_{b}^{1},r_{b}^{1},\ldots,x_{b}^{n_{b}},r_{b}^{n_{b}})_{b=1,\ldots,B}\in M
\]
if and only if for all $b=1,\ldots,B$ and all $a=1,\ldots,n_{b}$,
we have $x_{b}^{a}\in A$, $r_{b}^{a}\in\left[\frac{\text{€}}{\text{m}^{2}}\right]$,
$x_{b}^{a_{1}}\ne x_{b}^{a_{2}}$ if $a_{1}\ne a_{2}$, and at least
one of the following conditions is satisfied
\begin{align}
r_{b}^{a} & =L_{b}(x_{b}^{a})-k_{b}(x_{b}^{a})=R(x_{b}^{a})\label{eq:VT}\\
r_{b}^{a} & =L_{b}(x_{b}^{a})-k_{b}(x_{b}^{a})<R(x_{b}^{a})\nonumber \\
r_{b}^{a} & <L_{b}(x_{b}^{a})-k_{b}(x_{b}^{a})=R(x_{b}^{a})\nonumber \\
r_{b}^{a} & <L_{b}(x_{b}^{a})-k_{b}(x_{b}^{a})<R(x_{b}^{a})\nonumber \\
r_{b}^{a} & >L_{b}(x_{b}^{a})-k_{b}(x_{b}^{a})=R(x_{b}^{a})\nonumber 
\end{align}
\begin{align}
r_{b}^{a}=R(x_{b}^{a}) & >L_{b}(x_{b}^{a})-k_{b}(x_{b}^{a})\label{eq:ineq2}\\
r_{b}^{a}>R(x_{b}^{a}) & >L_{b}(x_{b}^{a})-k_{b}(x_{b}^{a})\nonumber \\
R(x_{b}^{a})>r_{b}^{a} & >L_{b}(x_{b}^{a})-k_{b}(x_{b}^{a}).\nonumber 
\end{align}
We also use the simplified notation $(x,r)=(x_{b}^{1},r_{b}^{1},\ldots,x_{b}^{n_{b}},r_{b}^{n_{b}})_{b=1,\ldots,B}$
to denote a configuration. We will think at $(x_{b}^{a},r_{b}^{a})$
as two components of the proper state space of the interacting entity
$(a,b)\in\mathcal{P}$, i.e.~the company $a=1,\ldots,n_{b}$ that
produces the commodity $b=1,\ldots,B$.
\end{defn}

\noindent The first one \eqref{eq:VT} of these conditions is called
\emph{von Thünen configuration}. Each one of these, except the von
Thünen one, corresponds to a possible configuration of an economy
where at least one of its agents is paying a cost or is loosing a
profit:
\begin{defn}
\label{def:costs1}Let $b=1,\ldots,B$, $a=1,\ldots,n_{b}$ and $(x_{b}^{1},r_{b}^{1},\ldots,x_{b}^{n_{b}},r_{b}^{n_{b}})_{b=1,\ldots,B}\in M$,
then
\begin{enumerate}
\item The cost paid by the tenant/company $a$, which produces the commodity
$b$, is located in $x_{b}^{a}$ and pays the rent $r_{b}^{a}$ is
\[
c_{\text{t}}(x_{b}^{a},r_{b}^{a}):=\begin{cases}
r_{b}^{a}-\left[L_{b}(x_{b}^{a})-k_{b}(x_{b}^{a})\right] & \text{if }r_{b}^{a}>L_{b}(x_{b}^{a})-k_{b}(x_{b}^{a})\\
c_{0\text{t}} & \text{otherwise}.
\end{cases}
\]
The quantity $c_{0\text{t}}\ge0$ represents a minimum non avoidable
cost related to the rent of this location (e.g.~administrative cost)
and paid by the tenant. We assume $c_{0\text{t}}$ sufficiently small,
i.e.~satisfying
\begin{equation}
r_{b}^{a}>L_{b}(x_{b}^{a})-k_{b}(x_{b}^{a})\ \Rightarrow\ c_{0\text{t}}<c_{\text{t}}(x_{b}^{a},r_{b}^{a})\label{eq:c_0t-small}
\end{equation}
for all $a$, $b$ (note that this implies $c_{0\text{t}}\le c_{\text{t}}(x_{b}^{a},r_{b}^{a})$,
and if $c_{0\text{t}}=c_{\text{t}}(x_{b}^{a},r_{b}^{a})$, then necessarily
$r_{b}^{a}\le L_{b}(x_{b}^{a})-k_{b}(x_{b}^{a})$). Clearly, if $r_{b}^{a}>L_{b}(x_{b}^{a})-k_{b}(x_{b}^{a})$,
the company $a$ is paying a rent $r_{b}^{a}$ higher than the company's
gain $L_{b}(x_{b}^{a})-k_{b}(x_{b}^{a})$, and it is hence at a loss.
\item The loss of profit of the renter of the location $x_{b}^{a}$ at rent
$r_{b}^{a}$ for the production of the commodity $b$ is
\[
l_{\text{r}}^{1}(x_{b}^{a},r_{b}^{a}):=\begin{cases}
\left[L_{b}(x_{b}^{a})-k_{b}(x_{b}^{a})\right]-r_{b}^{a} & \text{if }r_{b}^{a}<L_{b}(x_{b}^{a})-k_{b}(x_{b}^{a})\\
l_{1\text{r}} & \text{otherwise}.
\end{cases}
\]
The quantity $l_{1r}\ge0$ represents a minimum non avoidable cost
related to the rent of this location (e.g.~some tax or the cost to
know the land value $L_{b}(x_{b}^{a})$ and the cost of life $k_{b}(x_{b}^{a})$)
and paid by the renter. We assume $l_{1\text{r}}$ sufficiently small,
i.e.
\begin{equation}
r_{b}^{a}<L_{b}(x_{b}^{a})-k_{b}(x_{b}^{a})\ \Rightarrow\ l_{1\text{r}}<l_{\text{r}}^{1}(x_{b}^{a},r_{b}^{a})\label{eq:l_1r-small}
\end{equation}
for all $a$, $b$ (as above, this yields $l_{1\text{r}}\le l_{\text{r}}^{1}(x_{b}^{a},r_{b}^{a})$,
and if $l_{1\text{r}}=l_{\text{r}}^{1}(x_{b}^{a},r_{b}^{a})$, then
$r_{b}^{a}\ge L_{b}(x_{b}^{a})-k_{b}(x_{b}^{a})$). If $r_{b}^{a}<L_{b}(x_{b}^{a})-k_{b}(x_{b}^{a})$,
the renter is asking a lower rent $r_{b}^{a}$ with respect to the
better value $L_{b}(x_{b}^{a})-k_{b}(x_{b}^{a})$.
\item The loss of profit of the renter of the location $x_{b}^{a}$ at rent
$r_{b}^{a}$ with respect to all the possible commodities is
\[
l_{\text{r}}^{2}(x_{b}^{a},r_{b}^{a}):=\begin{cases}
R(x_{b}^{a})-r_{b}^{a} & \text{if }r_{b}^{a}<R(x_{b}^{a})\\
l_{2\text{r}} & \text{otherwise}
\end{cases}
\]
The quantity $l_{2r}\ge0$ represents a minimum non avoidable cost
related to the rent of this location (e.g.~some tax or the cost to
know the ideal rent $R(x_{b}^{b})$) and paid by the renter. We assume
$l_{2\text{r}}$ sufficiently small, i.e.
\begin{equation}
r_{b}^{a}<R(x_{b}^{a})\ \Rightarrow\ l_{2\text{r}}<l_{\text{r}}^{2}(x_{b}^{a},r_{b}^{a})\label{eq:l_2r-small}
\end{equation}
for all $a$, $b$. Again: $l_{2\text{r}}\le l_{\text{r}}^{2}(x_{b}^{a},r_{b}^{a})$
and if $l_{2\text{r}}=l_{\text{r}}^{2}(x_{b}^{a},r_{b}^{a})$, then
$r_{b}^{a}\ge R(x_{b}^{a})$. If $r_{b}^{a}<R(x_{b}^{a})$, the renter
is asking a rent $r_{b}^{a}$ lower that the best possible one $R(x_{b}^{a})$.
\end{enumerate}
\end{defn}

\noindent To these costs and losses, we associate the following probabilities
and hence expected costs and losses
\begin{defn}
\label{def:probCostsLosses}Let $(x,r)\in M$, then
\begin{enumerate}
\item The cost of the tenant $C_{\text{t}}$ can be computed as
\begin{align}
C_{\text{t}}(x,r): & =\sum_{b=1}^{B}\sum_{a=1}^{n_{b}}p_{\text{t}}(x_{b}^{a},r_{b}^{a})\cdot c_{\text{t}}(x_{b}^{a},r_{b}^{a}),\label{eq:costVT1}
\end{align}
where $\left(p_{\text{t}}(x_{b}^{a},r_{b}^{a})\right)_{a,b}$ is \emph{an}y
fixed \emph{non-degenerate} probability distribution, i.e.~$p_{\text{t}}(x_{b}^{a},r_{b}^{a})>0$
for all $a$, $b$. Note that $C_{\text{t}}(x,r)\ge c_{0\text{t}}$
because $c_{0\text{t}}\le c_{\text{t}}(x_{b}^{a},r_{b}^{a})$ for
all $a$, $b$, and that the probability $\left(p_{\text{t}}(x_{b}^{a},r_{b}^{a})\right)_{a,b}$
may depend on the cost $C_{\text{t}}$. Moreover, because of \eqref{eq:c_0t-small}
and the non-degenerateness condition, if $C_{\text{t}}(x,r)=c_{0\text{t}}$,
then $c_{\text{t}}(x_{b}^{a},r_{b}^{a})=c_{0\text{t}}$ for all $a$,
$b$.
\item Let $j=1,2$, the losses of profits of the renter can be computed
as
\begin{align}
L_{\text{r}}^{j}(x,r): & =\sum_{b=1}^{B}\sum_{a=1}^{n_{b}}p_{\text{r}}^{j}(x_{b}^{a},r_{b}^{a})\cdot l_{\text{r}}^{j}(x_{b}^{a},r_{b}^{a}),\label{eq:costVT2}
\end{align}
where $\left(p_{\text{r}}^{j}(x_{b}^{a},r_{b}^{a})\right)_{a,b}$,
for $j=1,2$, are \emph{any} fixed \emph{non-degenerate} probability
distribution, i.e.~$p_{\text{r}}^{j}(x_{b}^{a},r_{b}^{a})>0$ for
all $a$, $b$. As above $L_{\text{r}}^{j}(x,r)\ge l_{j\text{r}}$,
and if $L_{\text{r}}^{j}(x,r)=l_{j\text{r}}$, then $l_{\text{r}}^{j}(x_{b}^{a},r_{b}^{a})=l_{j\text{r}}$
for all $a$, $b$, because of \eqref{eq:l_1r-small} and \eqref{eq:l_2r-small}
and the non-degenerateness condition.
\end{enumerate}
\end{defn}

\noindent Compare \eqref{eq:costVT1}, \eqref{eq:costVT2} with \eqref{eq:UnificationFinite}
to recognize that $P_{(x,r)}^{\text{t}}(a,b):=p_{\text{t}}(x_{b}^{a},r_{b}^{a})$
and $P_{(x,r)}^{j,\text{r}}(a,b):=p_{\text{r}}^{j}(x_{b}^{a},r_{b}^{a})$
are the probabilities to average the considered unification costs.

\noindent These expected costs and losses are minimum at a von Thünen
configuration: $r_{b}^{a}=L_{b}(x_{b}^{a})-k_{b}(x_{b}^{a})=R(x_{b}^{a})$
for all $a$, $b$. Vice versa, when a given configuration $(x_{b}^{1},r_{b}^{1},\ldots,x_{b}^{n_{b}},r_{b}^{n_{b}})_{b=1,\ldots,B}\in M$
minimize costs and losses, is it a von Thünen configuration? In general
this is false: let $A=\{y_{1},y_{2}\}$, $B=2$, with
\begin{align*}
L_{1}(y_{j})-k_{1}(y_{j}) & =1\frac{\text{€}}{\text{m}^{2}}\\
L_{2}(y_{j})-k_{2}(y_{j}) & =2\frac{\text{€}}{\text{m}^{2}}.
\end{align*}
That is the commodity $b_{1}=1$ has a net land value equal to 1$\frac{\text{€}}{\text{m}^{2}}$
in both locations $y_{1}$, $y_{2}$, whereas the commodity $b_{2}=2$
has a double net land value in both locations. Assume, by contradiction,
that $(x_{1}^{1},r_{1}^{1},x_{2}^{2},r_{2}^{2})\in M$ minimizes costs
and losses. We would have
\begin{align}
c_{\text{t}}(x_{1}^{1},r_{1}^{1}) & =c_{0\text{t}}\iff r_{1}^{1}\le L_{1}(x_{1}^{1})-k_{1}(x_{1}^{1})=1\frac{\text{€}}{\text{m}^{2}}\nonumber \\
l_{r}^{1}(x_{1}^{1},r_{1}^{1}) & =l_{1\text{r}}\iff r_{1}^{1}\ge L_{1}(x_{1}^{1})-k_{1}(x_{1}^{1})=1\frac{\text{€}}{\text{m}^{2}}\nonumber \\
l_{r}^{2}(x_{1}^{1},r_{1}^{1}) & =l_{2\text{r}}\iff r_{1}^{1}\ge R(x_{1}^{1})=2\frac{\text{€}}{\text{m}^{2}}.\label{eq:lossNoVT}
\end{align}
This contradiction proves that, in this system, no configuration minimizes
costs and losses to the minimum values $c_{0\text{t}}$, $l_{1\text{r}}$,
$l_{2\text{r}}$. In other words, this system does not allow for a
von Thünen configuration. The problem here is that the system is lacking
in possible configurations. This justifies the main assumption \eqref{eq:existGoodLocations}
of the following
\begin{thm}
\label{thm:VTconf}If any company can be situated in a good location:
\begin{equation}
\forall b=1,\ldots,B\,\forall a=1,\ldots,n_{b}\,\exists y\in A:\ y\text{ is good for }b,\label{eq:existGoodLocations}
\end{equation}
then a given configuration $(x,r)\in M$ minimizes the average costs
and losses
\begin{align}
C_{\text{t}}(x,r) & \le C_{\text{t}}(y,s)\nonumber \\
L_{\text{r}}^{1}(x,r) & \le L_{\text{r}}^{1}(y,s)\quad\forall(y,s)\in M\label{eq:minim}\\
L_{\text{r}}^{2}(x,r) & \le L_{\text{r}}^{2}(y,s)\nonumber 
\end{align}
if and only if $(x,r)$ is a von Thünen configuration, i.e.
\[
r_{b}^{a}=L_{b}(x_{b}^{a})-k_{b}(x_{b}^{a})=R(x_{b}^{a})\quad\forall a,b.
\]
The same result holds in any IS having a value function $L_{b}:A\ra\R$
and a cost function $k_{b}:A\ra\R$ for each interacting entities
$(a,b)\in\mathcal{P}$ in a finite population, and for the quantities
as defined in Def.~\ref{def:probCostsLosses}.
\end{thm}

\begin{proof}
If $(x,r)$ is a von Thünen configuration, then $C_{\text{t}}(x,r)=c_{0\text{t}}\le C_{t}(y,s)$,
$L_{\text{r}}^{1}(x,r)=l_{1\text{r}}\le L_{\text{r}}^{1}(y,s)$, $L_{\text{r}}^{2}(x,r)=l_{2\text{r}}\le L_{\text{r}}^{2}(y,s)$
by Def.~\ref{def:costs1}, so that \eqref{eq:minim} hold. Vice versa,
using assumption \eqref{eq:existGoodLocations}, we can construct
a von Thünen configuration $(y,s)$ by choosing a good location for
every commodity:
\[
\forall b,a\,\exists y_{b}^{a}\in A:\ R(y_{b}^{a})=L_{b}(y_{b}^{a})-k_{b}(y_{b}^{a}).
\]
But \eqref{eq:minim} yields
\begin{align*}
c_{0\text{t}} & \le C_{\text{t}}(x,r)\le C_{\text{t}}(y,s)=c_{0\text{t}}\\
l_{1\text{r}} & \le L_{\text{r}}^{1}(x,r)\le L_{\text{r}}^{1}(y,s)=l_{1\text{r}}\\
l_{2\text{r}} & \le L_{\text{r}}^{2}(x,r)\le L_{\text{r}}^{2}(y,s)=l_{2\text{r}}.
\end{align*}
Therefore, $r_{b}^{a}=L_{b}(x_{b}^{a})-k_{b}(x_{b}^{a})\ge R(x_{b}^{a})$
because of Def.~\ref{def:costs1} and Def.~\ref{def:probCostsLosses}.
Therefore, also $(x,r)$ is necessarily a von Thünen configuration
by \eqref{eq:VT}.
\end{proof}
Clearly, assumption \eqref{eq:existGoodLocations} is not completely
realistic in real-world economies. Once again, this underscores that
what cost functions to consider in this kind of economic models, is
a modeling/\-philosophical/political choice. Depending on our political
choices, other types of costs can be considered, such as: environmental
costs, energy consumption, loss of profits, state's costs due to loss
of job places because of transfer of branches, company's stock price,
etc. One can also argue that also the maximization of profit is unrealistic
because it implies the full knowledge of the entire system and unrealistic
tendencies, such as that of trying the renting of centered locations
to jewelries as soon as the demand of jewels increases.

In another political approach, we can denote by $\text{tax}(x_{b}^{a})$
the real estate tax that must be paid by the owner (renter) of the
real estate property located (and rented) at $x_{b}^{a}$; then, we
can consider the previous cost of the tenant $C_{\text{t}}$ and,
instead of the losses of profits $L_{\text{r}}^{j}$, and examine
the cost
\begin{equation}
c_{\text{r}}(x_{b}^{a},r_{b}^{a}):=\begin{cases}
\text{tax}(x_{b}^{a})-r_{b}^{a} & \text{if }r_{b}^{a}<\text{tax}(x_{b}^{a})\\
0 & \text{otherwise},
\end{cases}\label{eq:costRenter}
\end{equation}
and the corresponding average cost $C_{\text{r}}$. In this model,
any configuration where $r_{b}^{a}\le L_{b}(x_{b}^{a})-k_{b}(x_{b}^{a})$
and $r_{b}^{a}\ge\text{tax}(x_{b}^{a})$ minimizes the costs $C_{\text{t}}$
and $C_{\text{r}}$. We have hence a more realistic model without
any assumption of profits maximization or of full information about
the market.

\subsubsection{\label{subsec:Maximization-of-diversification}Maximization of diversification
forces}

A very natural flux of goods is the amount of commodity $b=1,\ldots,B$
that each company $a=1,\ldots,n_{b}$ is selling to the market $x_{m}$.
We therefore want to see what the maximization of the diversification
forces yields if we consider these interactions. We think at a generic
IS where interacting entities are $E=\{(a,b)\mid b=1,\ldots,B,\ a=1,\ldots,n_{b}\}\cup\{x_{m}\}$,
even if we still keep a language of land use theory. The population
we want to consider is hence $\mathcal{P}_{b}=\left\{ (a,b)\in E\mid a=1,\ldots,n_{b}\right\} $
for each fixed $b=1,\ldots,B$, i.e.~the collection of the aforementioned
companies, so that its cardinality is $\left|\mathcal{P}_{b}\right|=n_{b}$.
We can therefore think at ``selling'' interactions $i_{b}^{a}:(a,b)\xra{r_{b}^{a},\text{s}}x_{m}$
having the company $(a,b)$ as agent and the market $x_{m}$ as patient,
and hence we consider $I_{\mathcal{P}_{b}}:=\{i_{b}^{a}\mid a=1,\ldots,n_{b}\}$.
The good of $i_{b}^{a}$ is a quantity $\phi_{b}^{a}\in(0,Q_{b}]=:R_{i_{b}^{a}}$
belonging to the space of resources $R_{i_{b}^{a}}$ of the interaction
$i_{b}^{a}$. Using an example only to understand, in a simpler deterministic
model of constant production, if $s_{b}^{a}$ is the surface of the
company $a$ used for the production of $1\textsf{ year}$ of $b$
(more generally, the amount of units of production $v_{b}$ used in
$1\textsf{ year}$, see definition \ref{enu:yield} in Sec.~\ref{subsec:VT-impedanceZones}),
then we have
\begin{align*}
\phi_{b}^{a} & =y_{b}(x_{b}^{a})\cdot s_{b}^{a}\in\left[\frac{u_{b}}{\textsf{year}}\right]\\
\Delta_{b} & =\sum_{a=1}^{n_{b}}\phi_{b}^{a}.
\end{align*}
We can thus think at $\Delta_{b}$ as the amount of commodity $b$
demanded by $x_{m}$ in $1\textsf{ year}$.

In general, i.e.~independently from this particular example of land
use theory, we consider a global state space $\bar{M}_{\mathcal{P}_{b},I_{\mathcal{P}_{b}}}$
such that
\begin{equation}
R_{\mathcal{P}_{b}}=\left\{ (\phi_{b}^{1},\ldots,\phi_{b}^{n_{b}})\in(0,\Delta_{b}]^{n_{b}}\mid\sum_{a=1}^{n_{b}}\phi_{b}^{a}=\Delta_{b}\right\} ,\label{eq:globalResources}
\end{equation}
and the diversification probability corresponding to a Bernoulli process
\begin{equation}
Q_{\gamma}(i_{b}^{a}):=\frac{\gamma^{a}}{\Delta_{b}}=\frac{\phi_{b}^{a}}{\Delta_{b}}=:q_{b}^{a}\quad\forall\gamma=(\phi_{b}^{1},\ldots,\phi_{b}^{n_{b}})\in R_{\mathcal{P}_{b}}.\label{eq:diversProbVT}
\end{equation}
In other words, the probability to extract a unit of commodity $b$
produced by the company $a$, among all those flowing to the market
$x_{m}$ in one year, equals the fraction $\frac{\phi_{b}^{a}}{\Delta_{b}}$
of goods $\phi_{b}^{a}$ produced by $a$ over the demanded total
$\Delta_{b}$. 

Assume now that the system is in an emergent pattern state $\gamma=(\phi_{b}^{1},\ldots,\phi_{b}^{n_{b}})$.
We have that
\begin{align}
D_{I_{\mathcal{P}_{b}}}(\gamma) & =:D_{I_{\mathcal{P}_{b}}}(\phi_{b}^{1},\ldots,\phi_{b}^{n_{b}})=-\sum_{i\in I_{\mathcal{P}_{b}}}Q_{\gamma}(i)\cdot\log_{2}Q_{\gamma}(i)\label{eq:maxDivers}\\
 & =-\sum_{a=1}^{n_{b}}q_{b}^{a}\cdot\log_{2}q_{b}^{a}\ge D_{I_{\mathcal{P}_{b}}}(s)\quad\forall s.
\end{align}
The necessary result yielded by this maximization is well known, and
it can be briefly repeated here: using the Lagrange's multiplier method,
we get
\begin{multline*}
\exists\lambda\in\R:\ \left.\frac{\partial D_{I_{\mathcal{P}_{b}}}(\gamma^{1},\ldots,\gamma^{n_{b}})}{\partial\gamma^{a}}\right|_{\gamma^{(-)}=\phi_{b}^{(-)}}=\\
=\frac{\partial}{\partial\gamma^{a}}\left.\left(-\sum_{a=1}^{n_{b}}\frac{\gamma^{a}}{\Delta_{b}}\cdot\log_{2}\frac{\gamma^{a}}{\Delta_{b}}\right)\right|_{\gamma^{(-)}=\phi_{b}^{(-)}}=\\
=\lambda\cdot\frac{\partial}{\partial\gamma^{a}}\left.\left(\sum_{a=1}^{n_{b}}\frac{\gamma^{a}}{\Delta_{b}}-1\right)\right|_{\gamma^{(-)}=\phi_{b}^{(-)}}.
\end{multline*}

\noindent This gives $\phi_{b}^{a}=\frac{2^{-\lambda}}{e\Delta_{b}}$,
so that $\Delta_{b}=\sum_{a=1}^{n_{b}}\phi_{b}^{a}=n_{b}\cdot\frac{2^{-\lambda}}{e\Delta_{b}}$
and
\begin{equation}
\phi_{b}^{a}=\frac{\Delta_{b}}{n_{b}}.\label{eq:equiFluxes}
\end{equation}
We can state this result as a general
\begin{thm}
\label{thm:VT-GEP}Let $\mathcal{I}$ be an IS with interacting entities
$E=\{(a,b)\mid b=1,\ldots,B,\ a=1,\ldots,n_{b}\}\cup\{x_{m}\}$. Consider
the population $\mathcal{P}_{b}=\left\{ (a,b)\in E\mid a=1,\ldots,n_{b}\right\} $
for each fixed $b=1,\ldots,B$ and with adapting interactions $I_{\mathcal{P}_{b}}:=\{i_{b}^{a}\mid a=1,\ldots,n_{b}\}$,
where $i_{b}^{a}:(a,b)\xra{r_{b}^{a},\text{s}}x_{m}$, with global
space of resources given by \eqref{eq:globalResources} and diversification
Bernoulli probabilities \eqref{eq:diversProbVT}. Then $\gamma=(\phi_{b}^{1},\ldots,\phi_{b}^{n_{b}})\in R_{\mathcal{P}_{b}}$
is a pattern with maximum diversification $D_{I_{\mathcal{P}_{b}}}$
if and only if $\phi_{b}^{a}=\frac{\Delta_{b}}{n_{b}}$ holds for
all $(a,b)\in E$.
\end{thm}

\noindent In land use theory, fluxes of commodities $b$ are hence
equally divided among all companies. We can interpret this property
as a resilience characteristic of the adapted population $\mathcal{P}_{b}$,
due to the absence of monopolies, maximal distribution of work, etc.
We can hence say that in this case the GEP coincides with a well studied
property of stable economies. In other words, if the population $\mathcal{P}_{b}$
follows the GEP, then the companies try to move in different configuration
$(x,r)\in M$ so that to decrease the average costs of Def.~\ref{def:costs1}
and, at the same time, thanks to some forces not explicitly represented
in the model (e.g.~a suitable taxation system or a clever use of
resources or a population tendency to be resilient), they also evolve
so that to satisfy \eqref{eq:equiFluxes}, i.e.~the absence of monopolies
and a maximal \emph{distribution} of work and resources.

Equation \eqref{eq:equiFluxes} can also be obtained from Thm.~\ref{thm:powerLaw}
and considering costs such as those of Def.~\ref{def:costs1} or
\eqref{eq:costRenter} (in general any cost function $E_{\mathcal{P}}$
that does not depend on the probabilities, so that the corresponding
derivative in \eqref{eq:hpParC} is zero). In Thm.~\ref{thm:powerLaw},
we can set $d=n_{b}$ and $x_{j}=\frac{\phi_{b}^{j}}{\Delta_{b}}$,
noting therefore that the costs do not depend on these probability
but only on the configuration $(x,r)$. Assuming \eqref{eq:maxDivers}
(or even only that the ratio $\frac{E_{\mathcal{P}_{b}}}{D_{I_{\mathcal{P}_{b}}}}$
is minimum) and setting $\alpha_{k}=0$, we have $\sum_{k=1}^{n_{b}}k^{-\alpha_{k}\cdot\frac{D_{I_{\mathcal{P}_{b}}}}{E_{\mathcal{P}}}}=n_{b}$.
Therefore, if at least $n_{b}\ge3>e$, Thm.~\ref{thm:powerLaw} yields
$q_{k}(y)=p_{b}^{k}(y)=p_{b}^{1}(y)=\frac{1}{n_{b}}=\frac{\phi_{b}^{a}}{\Delta_{b}}$.
This simple example shows how the general Thm.~\ref{thm:powerLaw}
can be applied to completely different settings.

It is clear that the result \eqref{eq:equiFluxes} does not correspond
to a realistic behavior in real-world economies: However, it gives
elements to start thinking at the GEP as a way to measure how far
our economy is from a stable, adaptive, resilient system.

\section{Conclusions and future developments}

Even if it does not agree with a purely formalistic point of view
in philosophy of mathematics, a mathematical definition must be validated
as well. This validation ranges from useful and general theorems linked
to this definition, to the inclusion of several interesting examples,
in our case examples of CAS. It is therefore a long process performed
by the interested scientific community. On the one hand, we followed
the ideas of G.K.~Zipf, \cite{Zip49}, which are nowadays informally
frequently used in different modeling of CS, see e.g.~\cite{Bak97,Cru12,Gab99,Hak04,Man53,Mi-Pa07,JaBaCr,Ni-Pr77,Sor06}
and references therein. The GEP can be explained and used even only
at an intuitive level and the corresponding formalization is simple
and corresponding to the intuition. This is already a form of validation.
On the other hand, the present paper is only the first step in this
validation process. For example, both the power law Thm.~\ref{thm:powerLaw},
or Thm.~\ref{thm:VTconf} and the results of Sec.~\ref{subsec:Maximization-of-diversification}
about von Thünen-like models of CS, viewed as sufficiently general
mathematical results applicable to large classes of CS, move in this
direction. Clearly, a better von Thünen's model (e.g.~dynamical,
stochastic also in the movement to different configurations $(x,r)$,
with explicit modeling of forces that allow the population to approach
the ideal relation \eqref{eq:equiFluxes}, with more than one markets,
where the strong assumption \eqref{eq:existGoodLocations} does not
hold, etc.), or the applications of the related theorems to different
CAS (e.g.~phyllotaxis?) are only some of possible improvements.

As far as we know, the GEP is the first mathematical definition of
CAS \emph{with a proved universal applicability} supported by the
mathematical embeddings of \cite{Gio24_IS}. Having a clear mathematical
notion could be of great advantage for the understanding and future
development of the notion of CAS. For example, already the intuitive
discussion to arrive at the GEP we had in Sec.~\ref{sec:Intuitive-description-GEP},
allowed us to recognize that the notion of CAS has to depend on validated
costs, suitable probabilities to average these costs and measure the
diversification forces, to identify a set of adaptive interactions,
etc. Even the notion of emergent pattern can be mathematically defined,
but can be reasonably considered less important than the general notion
of GEP, where real-world CAS may only adapt evolving from a given
state towards a better one.

\end{document}